\documentclass[a4paper,preprint]{elsarticle}

\journal{SCP}
\usepackage{url,arydshln}
\usepackage{my-bib}
\usepackage[shortlineseq,shortlogic]{my-math}

\def\ttimes{\mathbin{\text{\tt*}}}

\def\Sig{\mathcal{F}}
\def\Vars{\mathcal{V}}
\def\Terms{\mathcal{T}}
\def\Lex{\Const{lex}}

\def\DP{\Const{DP}}
\def\SN{\Const{SN}}
\def\Dom{\Const{Dom}}

\def\RR{{\mathcal{R}}}

\def\PP{\mathcal{P}}

\def\Coef{\mathit{sc}}
\def\CoefOf#1{\Coef(#1)}
\def\VCoef{\mathsf{vc}}
\def\VCoefOf#1{\VCoef(#1)}

\def\A{{\mathcal{A}}}
\def\Asum{{\mathcal{S}um}}
\def\AsumP{{\mathcal{S}um^+}}
\def\Amax{{\mathcal{M}ax}}
\def\Ams{{\mathcal{MS}um}}
\def\Apol{{\mathcal{P}ol}}
\def\Amp{{\mathcal{MP}ol}}
\def\Amat{{\mathcal{M}at}}

\def\REL{\sqsupset}
\def\GT{\succ}

\def\GE{\succeq}
\def\GS{\succsim}

\def\NGE{\nsucceq}

\def\PGT{>_\Sig}

\def\PGS{\gtrsim_\Sig}
\def\PSIM{\sim_\Sig}
\def\PNGT{\ngtr_\Sig}

\def\gs{\gtrsim}
\def\ngs{\mathrel{\SuperImpose\gtrsim{/}}}
\def\AGT{>_{\!\A}}
\def\ANGT{\ngtr_{\!\A}}

\def\AGS{\gs_\A}

\def\GSopt{\mathrel{\SuperImpose\succsim{\lower.5ex\hbox{\rm\tiny(\hskip 1.1em)}}}}
\def\gsopt{\mathrel{\SuperImpose\gtrsim{\lower.5ex\hbox{\rm\tiny(\hskip 1.1em)}}}}
\def\geopt{\mathrel{\SuperImpose\ge{\lower.5ex\hbox{\rm\tiny(\hskip 1.1em)}}}}

\makeatletter
\def\GEnotG{%
	\my@mathchoice{\mathrel{\SuperImpose\ge{\raise.3ex\hbox{$\my@style/$}}}}}

\makeatother

\def\KBO{\mathrm{KBO}}
\def\GKBO{\mathrm{GKBO}}
\def\TKBO{\mathrm{TKBO}}
\def\RPO{\mathrm{RPO}}
\def\POLO{\mathrm{POLO}}
\def\RPOLO{\mathrm{RPOLO}}
\def\LPO{\mathrm{LPO}}

\def\WPO{\mathrm{WPO}}
\def\WPOA{{\WPO(\A)}}
\def\WPO{\mathrm{WPO}}
\def\WPOA{{\WPO(\A)}}
\def\WPOsum{{\WPO(\Asum)}}
\def\WPOsumP{{\WPO(\AsumP)}}
\def\WPOmax{{\WPO(\Amax)}}
\def\WPOms{{\WPO(\Ams)}}
\def\WPOpol{{\WPO(\Apol)}}
\def\WPOmp{{\WPO(\Amp)}}
\def\WPOmat{{\WPO(\Amat)}}

\def\Seq#1{%
	\bgroup
		\def\o##1{\overline{##1\hspace{1pt}}\hspace{-1pt}}%
		\hspace{1pt}\o#1
	\egroup
}

\def\AppPerm#1#2#3{[\Seq{{#2}_{#1}}]}
\def\AppPermSubst#1#2#3#4{[\SeqSubst{#2}_{#1}{#4}]}

\def\Weight{w}
\def\WeightOf#1{\Weight\Br{#1}}
\def\Wzero{\Weight_0}

\def\Max{\mathsf{max}}
\def\Pol{\mathsf{pol}}

\def\Wsum{\mathsf{w}}
\def\WsumOf#1{\Wsum(#1)}
\def\Wmax{\mathsf{w}}
\def\WmaxOf#1{\Wmax(#1)}
\def\XW{\overline{\mathsf{w}}}
\def\XWof#1{\XW(#1)}
\def\XWsm{\XW^\Wstatus}
\def\XWsmOf#1{\XWsm(#1)}
\def\Vec#1{\textbf{\textit{#1}}}
\def\VecWeight{\Vec{w}}
\def\VecWeightOf#1{\VecWeight(#1)}
\def\Cmat{SC}
\def\CmatOf#1{\Cmat(#1)}

\def\Wstatus{\mathit{ws}}
\def\WstatusOf#1{\Wstatus(#1)}

\def\Pen{\mathit{sp}}
\def\PenOf#1{\Pen(#1)}

\newcommand\SigD{\mathcal{D}}
\newcommand\Root{\mathsf{root}}

\newrefformat{item}{(\ref{#1})}

\makeatletter

\renewcommand{\p@enumii}{\theenumi}
\renewcommand{\p@enumiii}{\theenumi\theenumii--}

\makeatother

\def\m#1{\mathsf{#1}}
\def\ca{\mathsf{a}}
\def\cb{\mathsf{b}}

\def\f{\mathsf{f}}
\def\g{\mathsf{g}}
\def\h{\mathsf{h}}
\newcommand\ci{\mathsf{i}}
\def\F{\f^\sharp}

\def\s{\mathsf{s}}
\def\p{\mathsf{p}}

\def\PrecVarOf#1{\mathsf{p}_{#1}}
\def\WeiVarOf#1{\mathsf{w}_{#1}}
\def\WeiVarZero{\mathsf{w}_0}
\def\PenVarOf#1{\mathsf{sp}_{#1}}
\def\PermedVarOf#1{\mathsf{st}_{#1}}
\def\PermVarOf#1{\mathsf{st}_{#1}}
\def\CoefVarOf#1{\mathsf{sc}_{#1}}

\begin{document}

\begin{frontmatter}

\title{A Unified Ordering for Termination Proving%
\cut{\footnote{%
	The preliminary version of this paper, reformatted in the
	\texttt{elsarticle} style, is available at
	\url{http://www.trs.cm.is.nagoya-u.ac.jp/papers/SCP2014/reformatted.pdf}
}}}

\author[NU]{Akihisa Yamada\corref{cor1}}
\ead{ayamada@trs.cm.is.nagoya-u.ac.jp}
\cortext[cor1]{Corresponding author}

\author[GU]{Keiichirou Kusakari}
\ead{kusakari@gifu-u.ac.jp}

\author[NU]{Toshiki Sakabe}
\ead{sakabe@is.nagoya-u.ac.jp}

\address[NU]{Graduate School of Information Science, Nagoya University, Japan}
\address[GU]{Faculty of Engineering, Gifu University, Japan}

\begin{abstract}
We introduce a reduction order called the weighted path order (WPO)
that \REV{subsumes}{encompasses} many existing reduction orders.
WPO compares weights of terms as in the Knuth-Bendix order (KBO),
while WPO allows weights to be computed by 
\REV{%
	a wide class of interpretations.
}{%
	an arbitrary interpretation which is weakly monotone and weakly simple.
}%
We investigate summations, polynomials and maximums for such interpretations.
We show that KBO is a restricted case of WPO induced by summations,
the polynomial order (POLO) is subsumed by WPO induced by polynomials,
and the lexicographic path order (LPO) is a restricted case of WPO
induced by maximums.
By combining these interpretations,
we obtain an instance of WPO that unifies KBO, LPO and POLO.
In order to fit WPO in the modern dependency pair framework,
we further 
\REV{%
	provide a reduction pair based on WPO and
}{%
	extend WPO as a reduction pair by introducing
}%
partial statuses.
As a reduction pair, WPO also subsumes matrix interpretations.
We finally present SMT encodings of our techniques,
and demonstrate the significance of our work through experiments.
\end{abstract}


\end{frontmatter}

\section{Introduction}

Proving \emph{termination} of \emph{term rewrite systems (TRS\REV{s}{})}
is one of the most important task\REV{s}{} in
program verification and automated theorem proving, where
\emph{reduction orders} play a fundamental role.
The classic use of reduction orders in termination proving is 
illustrated in the following example:%
\REV{%

\example\label{ex:fact}
	Consider the following TRS $\RR_\m{fact}$:
	\[
		\RR_\m{fact} \DefEq 
		\left\{
			\begin{array}{rcl}
				\m{fact}(\m{0}) &\to& \s(\m{0})\\
				\m{fact}(\s(x)) &\to& \s(x) \ttimes \m{fact}(x)
			\end{array}
		\right.
	\]
	which defines the factorial function,
	provided the binary symbol $\ttimes$ is defined as multiplication.
}{%

	\[
		\RR_\m{half} \DefEq
		\begin{EqSet}
			\p(\s(x)) &\to x
		\\
			\m{half}(0) &\to 0
		\\
			\m{half}(\s(x)) &\to \s(\m{half}(\p(x)))
		\end{EqSet}
	\]
	The TRS $\RR_\m{half}$ defines the function $\m{half}$ that halves an input natural number.
}%
	We can prove termination of $\RR_{\REV{\m{fact}}{\m{half}}}$ by finding a reduction order $\GT$
	that satisfies the following constraints:%
\REV{%
	\begin{align*}
		\m{fact}(\m{0}) &\GT \s(\m{0})\\
		\m{fact}(\s(x)) &\GT \s(x) \ttimes \m{fact}(x)
	\end{align*}
	\endexample

}{%
	\begin{align*}
		\p(\s(x)) &\GT x
	\\
		\m{half}(0) &\GT 0
	\\
		\m{half}(\s(x)) &\GT \s(\m{half}(\p(x)))
	\end{align*}

}%
A number of reduction orders have been proposed,
and their efficient
\REV{%
	implementation is demonstrated by several
}{%
	implementations are proposed during
	the recent developments of
}%
automatic termination provers such as \AProVE \cite{GST06} or \TTTT \cite{KSZM09}.

One of the most well-known reduction orders is
the \emph{lexicographic path order (LPO)} of Kamin and L\'evy \cite{KL80},
a variant of the \emph{recursive path order (RPO)} of Dershowitz \cite{D82}.
LPO is unified with RPO using \emph{status} \cite{L83}.
Recently,
Codish \etal \cite{CGST12} proposed
an efficient implementation using a SAT solver for
termination proving by
RPO with status.

The \emph{Knuth-Bendix order (KBO)} \cite{KB70} is 
the \REV{oldest}{most historical} reduction order.
KBO has become a practical alternative in automatic termination checking since
Korovin and Voronkov \cite{KV03} discovered a polynomial-time algorithm
for termination proof\REV{s}{} with KBO.
Zankl \etal \cite{ZHM09} proposed another implementation method via SAT/SMT encoding,
and verified a significant improvement
in efficiency over dedicated implementations of
the polynomial-time algorithm.
However,
KBO is disadvantageous compared to LPO when \emph{duplicating} rules
(where a variable occurs more often in the right\REV{-}{ }hand side than in the left\REV{-}{ }hand side)
are considered.
Actually, no duplicating rule can be oriented by KBO.
To overcome this disadvantage,
Middeldorp and Zantema \cite{MZ97} proposed the \emph{generalized KBO (GKBO)},
which generalizes weights over
algebras that are weakly monotone and \emph{strictly simple}:
$f(\dots,x,\dots) > x$.
Ludwig and Waldmann proposed another extension of KBO called
the \emph{transfinite KBO (TKBO)} \cite{LW07,KMV11,WZM12},
which extends the weight function to allow linear polynomials over ordinals.
However, proving termination with TKBO involves \REV{solving the }{}satisfiability problem of
non-linear arithmetic which is undecidable in general.
Moreover, TKBO still does not \REV{subsume}{encompass} LPO.

The \emph{polynomial order (POLO)} of Lankford \cite{L75} interprets 
each function symbol by a strictly monotone polynomial.
Zantema \cite{Z01} extended the method \REV{to}{over} algebras
\REV{%
}{%
	that are weakly monotone and \emph{weakly} simple:
	$f(\dots,x,\dots) \ge x$,
}%
and suggested combining the ``max'' operator with polynomial interpretations
(\emph{max-polynomials} in terms of \cite{FGMSTZ08}).
Fuhs \etal proposed an efficient SAT encoding of POLO in \cite{FGMSTZ07}, and
a general version of POLO with max in \cite{FGMSTZ08}.

\REV{%
	The \emph{dependency pair (DP) method} of \citet{AG00}
}{%
	The \emph{dependency pair (DP) framework} \cite{AG00,HM05,GTSF06}
}%
significantly enhances the classic approach of reduction orders by
analyzing cyclic dependencies between rewrite rules.
In the DP \REV{method}{framework},
reduction orders are \REV{extended}{relaxed} to \emph{reduction pairs} $\Tp{\GS,\GT}$,
and it suffices if one rule in a \REV{recursive}{cyclic} dependency is strictly oriented,
and other rules are only weakly \REV{oriented}{so}.
\REV{%
	\example\label{ex:DP}\mbox{}%
}{}%
Consider again the TRS $\RR_{\REV{\m{fact}}{\m{half}}}$.
There is one cyclic dependency in $\RR_{\REV{\m{fact}}{\m{half}}}$,
that is represented by the \emph{dependency pair}
$\REV{\m{fact}}{\m{half}}^\sharp(\s(x)) \to \REV{\m{fact}}{\m{half}}^\sharp(\p(x))$%
\REV{, where $\REV{\m{fact}}{\m{half}}^\sharp$ is a fresh symbol. We
}{. Hence we
}%
can prove termination of $\RR_{\REV{\m{fact}}{\m{half}}}$ by finding
a reduction pair $\Tp{\GS,\GT}$ that satisfies the following constraints:%
\footnote{%
	The last two constraints can be removed by considering \emph{usable rules}
	\cite{AG00}.
}%
\REV{%
	\begin{align*}
		\m{fact}^\sharp(\s(x)) &\GT \m{fact}^\sharp(x)
	\\
		\m{fact}(\m0) &\GS \s(\m0)
	\\
		\m{fact}(\s(x)) &\GS \s(x) \ttimes \m{fact}(x)
	\end{align*}
\endexample
}{%
	\begin{align*}
		\m{half}^\sharp(\s(x)) &\GT \m{half}^\sharp(\p(x))
	\\
		\p(\s(x)) &\GS x
	\\
		\m{half}(0) &\GS 0
	\\
		\m{half}(\s(x)) &\GS \s(\m{half}(\p(x)))
	\end{align*}

}%
One of the typical methods for designing reduction pairs is
\emph{argument filtering} \cite{AG00},
which generates reduction pairs from arbitrary reduction orders.
Hence, \REV{reduction orders are}{a reduction order is} still an important subject to study
in modern termination proving.
\REV{}{\par}%
Another typical technique is generalizing interpretation methods
to \emph{weakly} monotone ones,
\eg allowing $0$ coefficients for polynomial interpretations \cite{AG00}.
Endrullis \etal \cite{EWZ08} extended polynomial interpretations
to \emph{matrix interpretations},
and presented \REV{their}{its} implementation via SAT encoding.
\REV{}{\par}%
More recently, Bofill \etal \cite{BBRR13} proposed a reduction pair called \emph{RPOLO},
which unifies standard POLO and RPO by choosing either
\emph{RPO-like} or \emph{POLO-like} comparison
depending on function symbols.

These reduction orders and reduction pairs
require different correctness proofs and different implementations.
In this paper, we extract the underlying essence of these reduction orders and
introduce a general reduction order
called the \emph{weighted path order (WPO)}.
Technically, WPO is a further generalization of GKBO that
relaxes the strict simplicity condition of weights to \emph{weak simplicity}.
This relaxation become\REV{s}{} possible
by combining \REV{the }{}recursive checks of LPO with GKBO.
While strict simplicity is so restrictive that
GKBO does not even subsume the standard KBO,
weak simplicity is so \REV{general}{generous} that
WPO \REV{subsumes}{encompasses} not only KBO but also most of the reduction orders
described above (LPO, TKBO, POLO and so on), except for matrix interpretations%
\REV{
	which are not weakly simple in general%
}{}.

\REV{%
	There exist several earlier works on generalizing existing reduction orders.
	The \emph{semantic path order (SPO)} of \citet{KL80} is a generalization of
	RPO where precedence comparison is generalized to
	an arbitrary well-founded order on terms.
	However, to prove termination by SPO
	users have to ensure monotonicity by themselves,
	even if the underlying well-founded order is monotone (\cf \cite{BFR00}).
	On the other hand,
	monotonicity of WPO is guaranteed.
	\citet{BFR00} propose a variant of SPO
	that ensures monotonicity by using an external monotonic order.
	As well as LPO or POLO,
	also WPO can be used as such an external order.
	The \emph{general path order (GPO)} \cite{DH95,G96} is
	a very general framework that many reduction orders are subsumed.
	Due to the generality, however,
	implementing GPO seems to be quite challenging.
	Indeed, we are not aware of any tool that implements GPO.

}{}%
Instances of WPO are characterized by how weights are computed.
In particular, we introduce the following instances of WPO
and investigate their relationships with existing reduction orders:
\begin{itemize}
\item
	$\WPOsum$ which uses summations for weight computation.
	KBO can be obtained as a restricted case of $\WPOsum$,
	where the \emph{admissibility} condition is enforced, and
	weights of constants must be greater than $0$.
	$\WPOsum$ is free from these restrictions, and
	we verify that each extension strictly increases the power of the order.
\item
	$\WPOpol$ which uses monotone polynomial interpretations for weight computation.
	As a reduction order,
	POLO is subsumed by $\WPOpol$.
	TKBO can be obtained as a restricted case of $\WPOpol$,
	where interpretations are linear polynomials,
	\REV{}{the }admissibility is enforced, and
	interpretations of constants are greater than $0$.
\item
	$\WPOmax$ which uses maximums for weight computation.
	LPO can be obtained as a restricted case of $\WPOmax$, 
	where the weights of all symbols are fixed to $0$.
	In order to keep the presentation simple, we omit \emph{multiset status}
	and only consider \emph{LPO with status}.
	Nonetheless, it is easy to extend this result \REV{to}{for} \emph{RPO with status}.
\item
	$\WPOmp$ which combines polynomial\REV{s}{} and maximum for interpretation,
	and its variant $\WPOms$ whose coefficients are fixed to $1$.
	WPO($\Ams$) generalizes KBO and LPO, and
	$\WPOmp$ moreover subsumes POLO (with max)
	as a reduction order.
\end{itemize}
Note that all the instances described above use weakly simple algebras
\REV{which cannot be used for GKBO}{and hence not possible by GKBO}.

Next we extend WPO \REV{to}{as} a reduction pair by incorporating
\emph{partial statuses} \cite{YKS13b}.
This extension further relaxes the weak simplicity condition,
and arbitrary weakly monotone interpretations can be
used for weight computation. Hence as
a reduction pair, WPO also \REV{subsumes}{encompasses the} matrix interpretations,
as well as KBO, TKBO, LPO and POLO.
Though RPOLO also unifies RPO and POLO,
we show that WPO and RPOLO are incomparable in \REV{general}{theory}.%
\footnote{%
	It is possible to consider an instance of WPO that uses RPOLO for weight computation.
}
Moreover in practice,
WPO \REV{brings}{shows} significant benefit \REV{on}{in} the problems
from the \emph{Termination Problem Data Base (TPDB)} \cite{TPDB13},
while (the first-order version of) RPOLO does not,
as reported in \cite{BBRR13}.

Finally, we present an efficient implementation using
\REV{}{the }state-of-the-art SMT solvers.
By extending \cite{ZHM09},
we present SMT encoding techniques for the instances of WPO introduced so far.
In particular, \REV{the }{}orientability problem\REV{s}{} of $\WPOsum$, $\WPOmax$ and $\WPOms$
are reduced to a satisfiability problem of linear arithmetic, which is known to be decidable.
Through experiments in TPDB problems,
we also verify the efficiency of our implementation
and significance of WPO in practice.

The remainder of this paper is organized as follows:
In \prettyref{sec:preliminaries}
we recall some basic notions of term rewriting and
the definitions of existing orders.
\prettyref{sec:WPO order} is devoted \REV{to}{for} WPO as a reduction order.
There we present the definition of WPO, and then
investigate several instances of WPO and
show their relationships with existing orders.
In \prettyref{sec:pair}
we extend WPO \REV{to}{as} a reduction pair.
We present the definition of the reduction pair
and a soundness proof in \prettyref{sec:pair definition}.
Then the definition is refined in \prettyref{sec:pair refinements} and
\REV{the relationship to}{relationships with} existing reduction pairs
\REV{is}{are} shown in \prettyref{sec:pair instances}.
\prettyref{sec:encodings} presents SMT encodings for 
the instances of WPO introduced so far,
and some implementation issues are discussed in \prettyref{sec:optimizations}.
In \prettyref{sec:experiments} we verify the significance of our work
through experiments and we conclude in \prettyref{sec:conclusion}.

A preliminary version of this paper appeared in \cite{YKS13}.
The results for \REV{the reduction orders of}{reduction order in}
\prettyref{sec:WPO order} are basically the same as in \cite{YKS13}\REV{}{,
though the \REV{the presentation is}{presentations are} refined}.
The definition of WPO as a reduction pair in \prettyref{sec:pair} is new, and
hence most of the results in \prettyref{sec:pair} are new.
The \REV{revised}{new} version of WPO further subsumes POLO and matrix interpretations
as a reduction pair\REV{}{, as shown \REV{by}{as} 
Corollaries \ref{cor:WPO>=POLO pair} and \ref{cor:WPO>=MAT pair}, \resp,
in contrast to \cite{YKS13}}.
We also conclude that WPO does not subsume RPOLO by \prettyref{ex:RPOLO-WPO},
which was left open in \cite{YKS13}.
\REV{The}{%
Moreover, the SMT encodings in \prettyref{sec:encodings} are revised to fit
\REV{}{to }the new definition, and
the new} experimental results in
Sections \ref{sec:experiments pair} and \ref{sec:experiments combination} show
a significant improvement due to
the new definition of WPO as a reduction pair.

\section{Preliminaries}\label{sec:preliminaries}

\emph{Term rewrite systems (TRSs)} model first-order functional programs.
We refer the readers to \cite{BN98,Terese} for details \REV{on}{of} rewriting,
and only briefly recall some important notions needed in this paper.

A \Def{signature} $\Sig$ is a finite set of function symbols associated with \REV{an }{}arity.
The set of $n$-ary symbols is denoted by $\Sig_n$.
A \Def{term} is either a variable $x \in \Vars$ or \REV{of the }{in }form
$f(s_1,\dots,s_n)$ where $f \in \Sig_n$ and each $s_i$ is a term.
Throughout the paper, we abbreviate a sequence $a_1,\dots,a_n$ by $\Seq{a_n}$.
The set of terms constructed from $\Sig$ and $\Vars$ is denoted by
$\Terms(\Sig,\Vars)$.
The set of variables occurring in a term $s$ is denoted by $\Var(s)$, and
the number of occurrences of a variable $x$ in $s$ is denoted by $|s|_x$.
\REV{}{\par}%
A TRS is a set $\RR$ of pairs of terms called \Def{rewrite rules}.
A rewrite rule, written $l \to r$ where
$l \notin \Vars$ and $\Var(l) \supseteq \Var(r)$,
indicates that an instance of $l$ should be rewritten to
\REV{the }{}corresponding instance of $r$.
The \Def{rewrite relation} $\to_\RR$ induced by $\RR$ is
\REV{%
	the least relation which includes $\RR$ and is
	monotonic and stable.
	Here,
}{%
	the monotonic stable closure of $\RR$, where
}%
a relation $\REL$ on terms is
\begin{itemize}
\item
	\Def{monotonic} iff $s \REL t$ implies $f(\dots,s,\dots) \REL f(\dots,t,\dots)$
	for every context $f(\dots,\Box,\dots)$, and
\item
	\Def{stable} iff $s \REL t$ implies $s\theta \REL t\theta$
	for every substitution $\theta$.
\end{itemize}
A TRS $\RR$ is \Def{terminating} iff
no infinite rewrite sequence $s_1 \to_\RR s_2 \to_\RR \dots$ exists.

\subsection{Reduction Orders}

A classic method for proving termination is to find a \emph{reduction order}:
A \Def{reduction order} is a well-founded order which is monotonic and stable.
We say an order $\GT$ \Def{orients} a TRS $\RR$ iff
$l \GT r$ for every rule $l \to r \in \RR$; in other words, $\RR \subseteq {\GT}$.
It is easy to see the following:

\begin{theorem}\label{thm:order}\REV{\cite{Z94}}{}
	A TRS is terminating iff
	it is oriented by a reduction order.\qed
\end{theorem}

Ensuring well-foundedness of a reduction order is often a non-trivial task.
\REV{The following}{Following}
is a well-known technique of Dershowitz \cite{D82} for ensuring well-foundedness
based on \emph{Kruskal's tree theorem}\REV{.}{:}
A \Def{simplification order} is a strict order $\GT$ on terms, which is
monotonic and stable and satisfies \REV{the }{}\Def{subterm property}:
$f(\dots,s,\dots) \GT s$.

\begin{theorem}\cite{D82}
	For a finite signature,
	a simplification order is a reduction order.\qed
\end{theorem}

\REV{%
	In the latter, we only consider finite signatures.
}{}%
In the remainder of this section, we recall several existing reduction orders.

\subsubsection{Lexicographic Path Order}

We consider LPO \cite{KL80} with quasi-precedence and status;
a \Def{quasi-precedence} $\PGS$ is a quasi-order 
\REV{(\ie, a reflexive and transitive relation)}{} on $\Sig$,
whose strict part, denoted by $\PGT$, is well-founded.
The equivalence part of $\PGS$ is denoted by $\PSIM$.
A \Def{status function} $\sigma$ assigns 
\REV{to }{}each function symbol $f \in \Sig_n$
a permutation $[\Seq{i_n}]$ of positions
in $\SetOf{1,\dots,n}$.
We denote the list $[s_{i_1},\dots,s_{i_n}]$ by
$\AppPerm{\sigma(f)}{s}{n}$ for $\sigma(f) = [\Seq{i_n}]$.
\REV{%
	A strict order $\GT$ (associated with a quasi-order $\GS$)
	is lifted on lists as follows:
	$[\Seq{s_n}] \GT^\Lex [\Seq{y_m}]$ iff
	there exists $k < n$ \st
	$x_i \GS y_i$ for each $i \in \SetOf{1,\dots,k}$ and
	either $k = m$ or $k < m$ and $x_{k+1} \GT y_{k+1}$.

}{}%
\begin{definition}\label{def:LPO}
	For a quasi-precedence $\PGS$,
	the \emph{lexicographic path order} $\GT_\LPO$
	with status $\sigma$ is recursively defined as follows:
	$s = f(\Seq{s_n}) \GT_\LPO t$ iff
	\makeatletter
	\renewcommand{\p@enumii}{\theenumi--}
	\makeatother
	\renewcommand\theenumi{\alph{enumi}}
	\renewcommand\labelenumi{(\theenumi)}
	\renewcommand\theenumii{\roman{enumii}}
	\renewcommand\labelenumii{\theenumii.}
	\begin{enumerate}
	\item\label{item:LPO-simp}
		$\ForSome{i \in \{1,\dots,n \}} s_i\GE_\LPO t$, or
	\item\label{item:LPO-args}
		$t=g(\Seq{t_m})$, $\ForAll{j\in \{ 1, \dots, m \}} s\GT_\LPO t_j$ and either
		\begin{enumerate}
		\item\label{item:LPO-prec}
			$f\PGT g$, or
		\item
			$f\PSIM g$ and 
			$\AppPerm{\sigma(f)}{s}{n} \GT_\LPO^\Lex
			 \AppPerm{\sigma(g)}{t}{m}$.
		\end{enumerate}
	\end{enumerate}
\end{definition}

\begin{theorem}\cite{KL80}
	$\GT_\LPO$ is a simplification order and hence a reduction order.\qed
\end{theorem}
\unskip
\REV{%
\begin{example}\label{ex:LPO}
	Termination of $\RR_\m{fact}$ from \prettyref{ex:fact} can be shown by LPO.
	Both rules are oriented by case \prettyref{item:LPO-prec}
	with a precedence \st $\m{fact} \PGT \s$ for the first rule and
	$\m{fact} \PGT {\ttimes}$ for the second rule.
\end{example}\unskip
}{}%
\subsubsection{Polynomial Interpretations}

We basically follow the abstract definitions of \cite{Z01,EWZ08}.
A \Def{well-founded $\Sig$-algebra} $\A$
is a quadruple $\Tp{A,\gs,>,\cdot_\A}$
of a \Def{carrier set} $A$,
a quasi-order $\gs$ on $A$,
a well-founded order $>$ on $A$ which is \emph{compatible} with $\gs$,
\REV{\ie,}{\ie}
${\gs}\circ{>}\circ{\gs} \subseteq {>}$, and
an \Def{interpretation} $f_\A : A^n \to A$ for each $f\in\Sig_n$.
$\A$ is \Def{strictly} (\Def{weakly}) \Def{monotone} iff
$a \REV{\gsopt}{>} b$ implies $f_\A(\dots,a,\dots) \gsopt f_\A(\dots,b,\dots)$, and
\Def{strictly} (\Def{weakly}) \Def{simple} iff
$f_\A(\dots,a,\dots) \gsopt a$ for every $f \in \Sig$.
The relations $\gs$ and $>$ are extended \REV{to}{on} terms as follows:
$s \gsopt_\A t$ iff 
$\widehat\alpha(s) \gsopt \widehat\alpha(t)$ holds
for \REV{every}{all} assignment\REV{}{s} $\alpha : \Vars \to A$\REV{,
where $\widehat\alpha : \Terms \to A$ is the homomorphic extension of $\alpha$.
}{and its homomorphic extension $\widehat\alpha$.

}%
A \Def{polynomial interpretation} $\Apol$ interprets
every function symbol $f \in \Sig$ as a polynomial $f_\Apol$.
The carrier set of $\Apol$ is
$\{ a \in \Nat \mid a \ge \Wzero \}$ for some $\Wzero \in \Nat$%
\REV{, and the}{. The}
orderings are the standard $\ge$ and $>$ on $\Nat$.
$\Apol$ induces a reduction order if it is strictly monotone;
in other words, all arguments have coefficients at least $1$.

\begin{theorem}\cite{MN70,L75}
	If $\Apol$ is strictly monotone, then
	$>_\Apol$ is a reduction order.\qed
\end{theorem}
\unskip
\REV{%
\begin{example}\label{ex:POLO}
	Termination of $\RR_\m{fact}$ of \prettyref{ex:fact} can be
	shown by the polynomial interpretation $\Apol$ defined as follows:
	\begin{align*}
		\m{fact}_\Apol(x) &= 2 x + 2&
		\m{0}_\Apol &= 0\\
		\s_\Apol(x) &= 2 x + 1&
		x \ttimes_\Apol y &= x + y
	\end{align*}
	The left- and right-hand sides of 
	the rule $\m{fact}(\m{0}) \to \s(\m{0})$
	are interpreted as $2$ and $1$, \resp, and
	those of the rule 
	$\m{fact}(\s(x)) \to \s(x) \ttimes \m{fact}(x)$
	are interpreted as $4x + 4$ and $4x + 3$, \resp.
\end{example}\unskip
}{}%
\subsubsection{Knuth-Bendix Order}

$\KBO$ \cite{KB70} is induced by a \REV{quasi-}{}precedence and
a \Def{weight function} $\Tp{\Weight,\Wzero}$, where
$\Weight : \Sig \to \Nat$ and $\Wzero \in \Nat$ \st
$\WeightOf c \ge \Wzero$ for every constant $c \in \Sig_0$.
The weight $\WeightOf{s}$ of a term $s$ is defined as follows:
\[
	\WeightOf{s} \DefEq
	\begin{Cases}
	\Wzero
	&\text{ if } s \in \Vars \\
	\WeightOf{f} + \displaystyle\sum_{i=1}^{n}\WeightOf{s_i}
	&\text{ if } s = f(\Seq{s_n})
	\end{Cases}
\]
The weight function $\Weight$ is said
\REV{to be }{}\Def{admissible} for $\PGS$ iff
every unary symbol $f \in \Sig_1$ with $\WeightOf{f} = 0$ is
\REV{\emph{greatest} \wrt $\PGS$, \ie, $f \PGS g$ for every $g \in \Sig$}{maximum \wrt $\PGS$}.
In this paper we also consider status for KBO \cite{S89}.

\begin{definition}\label{def:KBO}
	For a quasi-precedence $\PGS$
	and a weight function $\Tp{\Weight,\Wzero}$,
	the \Def{Knuth-Bendix order}
	$\GT_\KBO$ with status $\sigma$ is
	recursively defined as follows:
	$s = f(\Seq{s_n}) \GT_\KBO t$ iff
	$|s|_x \ge |t|_x$ for all $x \in \Vars$ and either
	\begin{enumerate}
	\item\label{item:KBO-gt}
		$\WeightOf{s} > \WeightOf{t}$, or
	\item\label{item:KBO-ge}
		$\WeightOf{s} = \WeightOf{t}$ and either
		\begin{enumerate}
		\item\label{item:KBO-simp}
			$s = f^k(t)$ and $t \in \Vars$ for some $k > 0$, or
		\item\label{item:KBO-args}
			$t = g(\Seq{t_m})$ and either
			\begin{enumerate}
			\item\label{item:KBO-prec}
				$f\PGT g$, or
			\item\label{item:KBO-mono}
				$f\PSIM g$ and
				$\AppPerm{\sigma(f)}{s}{n} \GT_\KBO^\Lex
				 \AppPerm{\sigma(g)}{t}{m}$.
			\end{enumerate}
		\end{enumerate}
	\end{enumerate}
\end{definition}
Here we follow \cite{ZHM09}, and the range of $\Weight$ is restricted to $\Nat$.
According to \cite{KV03}, this does not decrease the power of $\KBO$ for finite TRSs.
Note that we do not assume $\Wzero > 0$ in the definition.
This assumption, together with \REV{}{the }admissibility is required for $\KBO$ to be a simplification order.
For details of the following result, we refer \eg \REV{to }{}\cite[Theorem~5.4.20]{BN98}.

\begin{theorem}
	If $\Wzero > 0$ and $\Weight$ is admissible for $\PGS$, then
	$\GT_\KBO$ induced by $\Tp{\Weight,\Wzero}$ and $\PGS$ is a simplification order,
	and hence a reduction order.
	\qed
\end{theorem}

\REV{%
The \emph{variable condition}
``$|s|_x \ge |t|_x$ for all $x \in \Vars$''
is often said to be a major disadvantage of KBO.
Due to this condition, no duplicating rule can be oriented by KBO.
\begin{example}
	Termination of $\RR_\m{fact}$ of \prettyref{ex:fact} cannot be shown by KBO,
	since the variable $x$ of the rule
	$\m{fact}(\s(x)) \to \s(x) \ttimes \m{fact}(x)$
	violates the variable condition.
\end{example}\unskip
}{}%
\subsubsection{Transfinite KBO}

TKBO \cite{LW07,KMV11,WZM12} extends KBO by introducing
a \Def{subterm coefficient function} $\Coef$,
that assigns a positive integer%
\footnote{We do not use \Def{transfinite} coefficients, since
they do not add power when finite TRSs are considered \cite{WZM12}.
\REV{%
	We will still use the acronym TKBO to denote
	the finite variant.
}{}%
}
$\CoefOf{f,i}$ to each $f \in \Sig_n$ and $i \in \{ 1,\dots,n \}$.
For a weight function $\Tp{\Weight,\Wzero}$ and a subterm coefficient function $\Coef$,
\REV{the refined weight }{}$\WeightOf{s}$ is \REV{defined}{refined} as follows:
\[
	\WeightOf{s} \DefEq
	\begin{Cases}
		\Wzero
		&\text{ if } s \in \Vars
	\\
		\WeightOf{f} + 
		\displaystyle\sum_{i = 1}^{n}\CoefOf{f,i} \cdot \WeightOf{s_i}
		&\text{ if } s = f(\Seq{s_n})
	\end{Cases}
\]
The \Def{variable coefficient}
$\VCoefOf{x,s}$ of $x$ in $s$ is defined recursively as follows:
\[
	\VCoefOf{x,s} \DefEq
	\begin{Cases}
		1 &\text{if } x = s
	\\
		0 &\text{if } x \neq \REV{s}{y} \in \Vars
	\\
		\displaystyle\sum_{i=1}^{n}\CoefOf{f,i} \cdot \VCoefOf{x,s_i}
		&\text{if } s = f(\Seq{s_n})
	\end{Cases}
\]
Then the order $\GT_\TKBO$ is obtained from \prettyref{def:KBO} by replacing
$|\cdot|_x$ by $\VCoefOf{x,\cdot}$ and $w(\cdot)$ by 
\REV{its refined version above}{refined ones}.

\begin{theorem}\cite{LW07}
	If $w_0 > 0$ and $w$ is admissible for $\PGS$, then
	$\GT_\TKBO$ is a simplification order and hence a reduction order.\qed
\end{theorem}
\unskip
\REV{%
\begin{example}\label{ex:TKBO}
	Consider again the TRS $\RR_\m{fact}$ of \prettyref{ex:fact}.
	Consider the weight function $\Weight$ defined as follows:
	\begin{align*}
		\WeightOf{\m{0}} &= 1&
		\WeightOf{\m{s}} &= 0&
		\WeightOf{\m{fact}} &= 1&
		\WeightOf{\ttimes} &= 0
	\end{align*}
	and the subterm coefficient function $\Coef$ defined as follows:
	\begin{align*}
		\CoefOf{\s,1} = \CoefOf{\m{fact},1} &= 2&
		\CoefOf{\ttimes,1} = \CoefOf{\ttimes,2} &= 1
	\end{align*}
	Finally, consider an admissible precedence \st
	$\m{s} \PGT \m{fact} \PGT \ttimes$.
	The first rule $\m{fact}(\m{0}) \to \s(\m{0})$ of $\RR_\m{fact}$ is
	oriented by case \prettyref{item:KBO-gt}.
	For both sides of the second rule $\m{fact}(\s(x)) \to \s(x) \ttimes \m{fact}(x)$,
	the variable coefficient of $x$ is $4$ and the weight is $3$.
	Hence, the rule is oriented by case \prettyref{item:KBO-prec}.
\end{example}\unskip
}{}%
\subsubsection{Generalized Knuth-Bendix Order}

GKBO \cite{MZ97} uses a weakly monotone and \emph{strictly} simple algebra for weight computation.
In the following version of GKBO, we
extend \cite{MZ97} with quasi-order $\AGS$ and quasi-precedence $\PGS$,
and omit \emph{multiset status}.

\begin{definition}
	For a \REV{quasi-}{}precedence \REV{$\PGS$}{$\PGT$} and a well-founded $\Sig$-algebra $\A$,
	the \Def{generalized Knuth-Bendix order} $\GT_\GKBO$
	is recursively defined as follows: $s =f(\Seq{s_n}) \GT_\GKBO t$ iff
	\renewcommand\theenumii{\roman{enumii}}%
	\renewcommand\labelenumii{\theenumii.}%
	\begin{enumerate}
	\item\label{item:GKBO-gt}
		$s \AGT t$, or
	\item\label{item:GKBO-ge}
		$s \AGS t = g(\Seq{t_m})$ and either
		\begin{enumerate}
		\item\label{item:GKBO-prec}
			$f \PGT g$, or
		\item
			$f \PSIM g$ and
			$\AppPerm{\sigma(f)}{s}{n} \GT_\GKBO^\Lex
			 \AppPerm{\sigma(g)}{t}{m}$.
		\end{enumerate}
	\end{enumerate}
\end{definition}

\begin{theorem}
	\cite{MZ97}
	If $\A$ is weakly monotone and strictly simple, then
	$\GT_\GKBO$ is a simplification order and hence a reduction order.\qed
\end{theorem}

\REV{%
	One of the most important advantage of GKBO is
	that it admits weakly monotone interpretations such as $\max$.
	\begin{example}
		Termination of $\RR_\m{fact}$ of \prettyref{ex:fact} can be shown
		by GKBO induced by an algebra $\A$ on $\Nat$ with interpretation \st
		\begin{align*}
		\m{fact}_\A(x) &= x + 2&
		\m{0}_\A &= 0\\
		\s_\A(x) &= x + 1&
		x \ttimes_\A y &= \max\{x,y\} + 1
		\end{align*}
		and a precedence \st $\m{fact} \PGT {\ttimes}$.
		The first rule $\m{fact}(\m{0}) \to \s(\m{0})$ is oriented by case \prettyref{item:GKBO-gt}.
		The second rule $\m{fact}(\s(x)) \to \s(x) \ttimes \m{fact}(x)$ is oriented by
		case \prettyref{item:GKBO-prec},
		since $x + 3 \ge \max\{x + 1, x + 2 \} + 1$.
	\end{example}

	Note that the strict simplicity condition is crucial for GKBO
	to be well-founded.
	If we modify the interpretation of the above example
	by $x \ttimes_\A y = \max\{x,y\}$,
	then GKBO will admit the following infinite sequence:
}{}%
	\[
\REV{%
		\m{fact}(\m{0}) \GT_\GKBO
		\m{0} \ttimes \m{fact}(\m{0}) \GT_\GKBO
		\m{0} \ttimes (\m{0} \ttimes \m{fact}(\m{0})) \GT_\GKBO \dots
}{}%
	\]
\subsection{The Dependency Pair Framework and Reduction Pairs}

The
\REV{%
	\emph{dependency pair (DP) method} \cite{AG00}
}{%
	\emph{dependency pair (DP) framework} \cite{AG00,HM05,GTS04,GTSF06}
}%
significantly enhances the classical method of reduction orders
by analyzing dependencies between rewrite rules.
We briefly recall the essential notions for
\REV{%
	its successor, the \emph{DP framework} \cite{HM05,GTS04,GTSF06}.

}{%
	the DP framework.
}%
Let $\RR$ be a TRS over a signature $\Sig$.
The \emph{root symbol} of a term $s = f(\Seq{s_n})$ is $f$ and denoted by
$\Root(s)$.
The set of \emph{defined symbols} \wrt $\RR$ is defined as
$\SigD \DefEq \{ \Root(l) \mid l \to r \in \RR \}$.
For each $f \in \SigD$, the signature $\Sig$ is extended by
a fresh \emph{marked symbol} $f^\sharp$ \REV{having the same arity}{whose arity is the same} as $f$.
For $s = f(\Seq{s_n})$ with $f \in \SigD$, the term
$f^\sharp(\Seq{s_n})$ is denoted by $s^\sharp$.
The set of \emph{dependency pairs} for $\RR$ is defined as
\(
	\DP(\RR) \DefEq 
	\{
		l^\sharp \to t^\sharp \mid l \to r \in \RR,
		t\text{ is a subterm of }r, \Root(t) \in \SigD
	\}
\).
A \emph{DP problem} is a pair $\Tp{\PP,\RR}$
of a TRS $\RR$ and a set $\PP$ of dependency pairs for $\RR$.
A DP problem $\Tp{\PP,\RR}$ is \emph{finite} iff $\to_\PP\cdot\to_\RR^*$ is well-founded,
where $\PP$ is viewed as a TRS.
The main result of the DP framework is the following:

\begin{theorem}\cite{AG00,GTSF06}
	A TRS $\RR$ is terminating 
	\REV{iff}{if} the DP problem $\Tp{\DP(\RR),\RR}$ is finite.
	\qed
\end{theorem}

Finiteness of a DP problem is proved by \emph{DP processors}:
A sound \emph{DP processor} 
\REV{gets a DP problem as input}{inputs a DP problem}
and outputs a set of (hopefully simpler) DP problems
\st the input problem is finite if all the output problems are finite.
Among other DP processors for transforming or simplifying DP problems
(\cf \cite{GTSF06} for a summary),
we recall the most important one:
A \emph{reduction pair} $\Tp{\GS,\GT}$ is a pair of relations on terms \st
$\GS$ is a monotonic and stable quasi-order, and
$\GT$ is a well-founded stable order which is \emph{compatible} with $\GS$,
\REV{\ie,}{\ie}
${\GS} \circ {\GT} \circ {\GS} \subseteq {\GT}$.

\begin{theorem}\cite{AG00,GTS04,HM05,GTSF06}\label{thm:reduction pair}
	Let $\Tp{\GS,\GT}$ be a reduction pair \st
	$\PP\cup\RR \subseteq {\GS}$ and $\PP' \subseteq {\GT}$.
	Then the DP processor that maps 
	$\Tp{\PP,\RR}$ to $\{ \Tp{\PP\setminus\PP',\RR} \}$
	is sound.
	\qed
\end{theorem}

\subsubsection{Weakly Monotone Interpretations}

To define a reduction pair by polynomial
interpretations, an interpretation
$\Apol$ need not be strictly monotone but only weakly \REV{monotone}{so};
in other words, $0$ coefficients are allowed.

\begin{theorem}\cite{AG00}
	If $\Apol$ is weakly monotone, then
	$\Tp{\ge_\Apol,>_\Apol}$ forms a reduction pair.\qed
\end{theorem}

Endrullis \etal \cite{EWZ08} extend\REV{}{s} linear polynomial interpretations
to \emph{matrix interpretations}.

\begin{definition}
	Given a fixed \emph{dimension} $d \in \Nat$,
	the well-founded algebra $\Amat$ consists of
	the carrier set $\Nat^d$ and
	the strict and quasi\REV{-}{ }orders on $\Nat^d$ defined as follows:
	\[
		\left(\begin{matrix}v_1\\\vdots\\v_d\end{matrix}\right)
		\gsopt
		\left(\begin{matrix}u_1\\\vdots\\u_d\end{matrix}\right)
		\DefIff
		v_1 \geopt u_1 \AND v_j \ge u_j \FORALL j \in \SetOf{2,\dots,d}
	\]
	The interpretation in $\Amat$ is induced by
	a function $\VecWeight$ that assigns
	a $d$-dimension vector $\VecWeightOf{f}$ to each $f \in \Sig$, and
	a function $\Cmat$ that assigns
	a $d \ttimes d$ matrix $\CmatOf{f,i}$ to each $f \in \Sig_n$ and
	$i \in \SetOf{1,\dots,n}$\REV{. It is}{, and} defined as follows:
	\[
		f_\Amat(\Seq{{\Vec{x}}_n}) =
		\VecWeightOf{f} + \sum_{i=1}^n \CmatOf{f,i}\cdot\Vec{x}_i
	\]
	The $i$-th row and $j$-th column element of a matrix $M$
	is denoted by $M^{i,j}$.
\end{definition}

\begin{theorem}\cite{EWZ08}
	For a matrix interpretation $\Amat$,
	$\Tp{\gs_\Amat,>_\Amat}$ forms a reduction pair.\qed
\end{theorem}

\subsubsection{Argument Filtering}

\emph{Argument filtering} \cite{AG00,KNT00} is a typical technique to design 
a reduction pair from a reduction order:
An \emph{argument filter} $\pi$ maps each $f \in \Sig_n$ to
either a position $i \in \{ 1,\dots, n \}$ or 
a list $[\Seq{i_m}]$ of positions \st $1 \le i_1 < \dots < i_m \le n$.
The signature $\Sig^\pi$ consists of
every $f \in \Sig$ \st $\pi(f) = [\Seq{i_m}]$,
and \REV{the }{}arity of $f$ is $m$ in $\Sig^\pi$.
An argument filter $\pi$ induces a mapping
$\pi : \Terms(\Sig,\Vars) \to \Terms(\Sig^\pi,\Vars)$ as follows:
\[
	\pi(s) \DefEq
	\begin{Cases}
		s			&\text{if } s \in \Vars
	\\
		\pi(s_i)	&\text{if } s = f(\Seq{s_n})\text{, }\pi(f) = i
	\\
		f(\pi(s_{i_1}),\dots,\pi(s_{i_m}))
		&\text{if }s = f(\Seq{s_n})\text{, }\pi(f) = [\Seq{i_m}]
	\end{Cases}
\]
For an argument filter $\pi$ and a reduction order $\GT$ on $\Terms(\Sig^\pi,\Vars)$,
the relations $\GS^\pi$ and $\GT^\pi$ on $\Terms(\Sig,\Vars)$ are defined as follows:
$s \GS^\pi t$ iff $\pi(s) \GE \pi(t)$, and 
$s \GT^\pi t$ iff $\pi(s) \GT \pi(t)$.

\begin{theorem}\cite{AG00}
	For a reduction order $\GT$ and an argument filter $\pi$,
	$\Tp{{\GS}^\pi,{\GT}^\pi}$ forms a reduction pair.
	\qed
\end{theorem}

\REV{%
The effect of argument filtering is especially apparent for KBO;
it relaxes the variable condition.
\begin{example}
	By applying an argument filter $\pi$ \st
	$\pi(\ttimes) = 1$ for the constraints in \prettyref{ex:DP},
	we obtain the following constraints:
	\[
		\m{fact}(\m{0}) \GS \s(\m{0})\qquad \m{fact}(\s(x)) \GS \s(x) \qquad
		\m{fact}^\sharp(\s(x)) \GT \m{fact}^\sharp(x)
	\]
	The first constraint can be satisfied by KBO with
	\eg $\WeightOf{\m{fact}} > \WeightOf{\s}$.
	The other constraints are satisfied by any instance of KBO.
\end{example}\unskip
}{}%
\section{WPO as a Reduction Order}\label{sec:WPO order}

In this section, we introduce a reduction order called
the \emph{weighted path order (WPO)} that further generalizes GKBO
by weakening the simplicity condition on algebras.
After showing some properties for the order,
we then introduce several instances of WPO by fixing algebras.
We investigate relationships between these instances of WPO and
existing reduction orders and show the potential of WPO.

We first introduce the definition of WPO.
In order to admit algebras that are not strictly but only weakly simple,
we employ the recursive checks
\REV{which ensure that LPO is}{that ensures LPO to be} a simplification order.

\def\WPOAS{{\WPO(\A,\sigma)}}
\begin{definition}[WPO]\label{def:WPO}
	For a quasi-precedence $\PGS$,
	a well-founded algebra $\A$ and a status $\sigma$,
	the \emph{weighted path order}
	$\GT_\WPOAS$
	is defined as follows:
	$s = f(\Seq{s_n}) \GT_\WPOAS t$ iff
	\begin{enumerate}
	\item
		$s \AGT t$, or
	\item
		$s \AGS t$ and
		\begin{enumerate}
		\item
			$\ForSome{i \in \SetOf{1,\dots,n}}s_i \GE_\WPOAS t$, or
		\item
			$t=g(\Seq{t_m})$,
			$\ForAll{j \in \SetOf{1,\dots,m}}s \GT_\WPOAS t_j$ and either
			\begin{enumerate}
			\item
				$f\PGT g$ or
			\item
				$f\PSIM g$ and 
				$\AppPerm{\sigma(f)}{s}{n} \GT_\WPOAS^\Lex
				 \AppPerm{\sigma(g)}{t}{m}$.
			\end{enumerate}
		\end{enumerate}
	\end{enumerate}
	We abbreviate $\WPOAS$ by $\WPOA$ and $\WPO$ when no confusion arises.
\end{definition}

Case \prettyref{item:WPO-gt} and the 
\REV{precondition of case}{ in} \prettyref{item:WPO-ge} are
the same as $\GKBO$.
Case \prettyref{item:WPO-simp} and the 
\REV{precondition of case}{condition in} \prettyref{item:WPO-args}
are the recursive checks that
correspond to \prettyref{item:LPO-simp} and \prettyref{item:LPO-args} of $\LPO$.
Note that here we may restrict \REV{}{\eg }$i$ in \prettyref{item:WPO-simp}
\REV{and $j$ in \prettyref{item:WPO-args}}{} to positions
\st $f(\dots,x_i,\dots) \ANGT x_i$, since otherwise we have $s \AGT t$,
which is \REV{covered by}{considered in} \prettyref{item:WPO-gt}.
Cases \prettyref{item:WPO-prec} and \prettyref{item:WPO-mono} are common among
$\WPO$, $\GKBO$ and $\LPO$.
\REV{%
In the appendix we prove the following soundness result:
\begin{theorem}\label{thm:WPO simple}
	If $\A$ is weakly monotone and weakly simple, then
	$\GT_\WPO$ is a simplification order and hence a reduction order.
\end{theorem}
}{}%

Note that \prettyref{thm:WPO simple} gives an alternative proof for
the following result of Zantema \cite{Z01}:
\begin{theorem}\cite{Z01}
	If a TRS $\RR$ is oriented by $\AGT$ for
	a weakly monotone and weakly simple algebra $\A$,
	then $\RR$ is simply terminating,
	\REV{\ie,}{\ie} its termination is shown by a simplification order.
\end{theorem}
\begin{proof}
	Since ${\AGT} \subseteq {\GT_\WPO}$, $\RR$ is oriented by
	the simplification order $\GT_\WPO$.
	\qedhere
\end{proof}

Moreover,
we can verify that $\WPO$ is a \REV{}{further }generalization of $\GKBO$.

\begin{theorem}\label{thm:WPO>=GKBO}
	If $\A$ is strictly simple, then
	${\GT_\GKBO} = {\GT_\WPO}$.
\end{theorem}
\begin{proof}
	The condition $s_i \GE_\WPO t$ of case \prettyref{item:WPO-simp}
	\REV{may be dropped}{is ignorable},
	since in that case we have $s \AGT s_i \AGS t$ by the assumption.
	Analogously,
	the condition $s \GT_\WPO t_j$ of case \prettyref{item:WPO-args} always holds,
	since $s \AGS t \AGT t_j$.
	Hence, case \prettyref{item:WPO-simp} and the condition in
	\prettyref{item:WPO-args} can be ignored, and
	the definition of $\WPO$ becomes equivalent to that of $\GKBO$.\qedhere
\end{proof}

\REV{%
	The relaxation of strict simplicity to weak simplicity is an important step.
	While GKBO does not even subsume the standard KBO,
	WPO subsumes not only KBO but also many other reduction orders.
}{}%
In the remainder of this section,
we investigate several instances of $\WPO$.

\subsection{WPO$(\Asum)$}

The first instance $\WPOsum$ is induced by an algebra $\Asum$, which
interprets function symbols as the summation operator $\sum$.
We obtain KBO as a restricted case of $\WPOsum$.
We design the algebra $\Asum$ from a weight function $\Tp{\Weight,\Wzero}$,
so that ${\GT_\WPOsum} = {\GT_\KBO}$ 
when $\Wzero > 0$ and \REV{}{the }admissibility is satisfied.

\begin{definition}
	The $\Sig$-algebra $\Asum$ induced by a weight function $\Tp{\Weight,\Wzero}$
	consists of the carrier set $\{ a \in \Nat \mid a \ge \Wzero \}$ and
	the interpretation which is defined as follows:
	\[
		f_\Asum(\Seq{a_n}) = \WeightOf{f} + \sum_{i=1}^{n}a_i
	\]
	If $\Wzero > 0$ is satisfied, we also write $\AsumP$ for $\Asum$.
\end{definition}

Obviously, $\Asum$ is strictly (and hence weakly) monotone and weakly simple.
We obtain the following as a corollary of \prettyref{thm:WPO simple}:

\begin{corollary}\label{cor:WPOsum}
	$\GT_\WPOsum$ is a reduction order.\qed
\end{corollary}

Now let us prove that $\GT_\KBO$ is obtained as a special case of $\GT_\WPOsumP$.
The following lemma verifies that $\Asum$ indeed works as the weight of KBO.

\begin{lemma}\label{lem:weight}
	$s \geopt_\Asum t$ iff
	$|s|_x \ge |t|_x$ for all $x \in \Vars$ and
	$\WeightOf{s} \geopt \WeightOf{t}$.
\end{lemma}
\begin{proof}
	The ``if'' direction is easy. For the ``only-if'' direction,
	suppose $s \geopt_\Asum t$.
	Define the assignment $\alpha_0$ which maps all variables to $\Wzero$.
	We have $\widehat\alpha_0(s) \geopt \widehat\alpha_0(t)$, that is
	$\WeightOf{s} \geopt \WeightOf{t}$.
	Furthermore, define the assignment $\alpha_x$ which maps
	$x$ to $\WeightOf{s} + \Wzero$ and \REV{other variables}{others} to $\Wzero$.
	We have $\widehat\alpha_x(s) \geopt \widehat\alpha_x(t)$, which implies
	$\WeightOf{s} + |s|_x \cdot \WeightOf{s} \geopt \WeightOf{t} + |t|_x \cdot \WeightOf{s}$.
\REV{%
	Here, $|s|_x < |t|_x$ cannot hold since $\WeightOf{t} \ge 0$.
	We conclude
}{%
	Hence, we get
}%
	$|s|_x \ge |t|_x$.
	\qedhere
\end{proof}

\begin{theorem}\label{thm:WPO>=KBO}
	If $\Wzero > 0$ and $\Weight$ is admissible for $\PGS$, then
	${\GT_\WPOsum} = {\GT_\KBO}$.
\end{theorem}
\begin{proof}
	For arbitrary terms $s = f(\Seq{s_n})$ and $t$,
	we show $s \GT_\WPOsum t$ iff $s \GT_\KBO t$ by induction on $|s| + |t|$.
\REV{%
	Because of the admissibility assumption,
	we may assume that $\GT_\KBO$ is a simplification order.
}{}%
	\begin{itemize}
	\item
		Suppose $s \GT_\KBO t$.
		If $\WeightOf{s} > \WeightOf{t}$, then we have
		$s >_\Asum t$ by \prettyref{lem:weight} and 
		$s \GT_\WPOsum t$ by \prettyref{item:WPO-gt} of \prettyref{def:WPO}.
		Let us consider that
		$\WeightOf{s} = \WeightOf{t}$.
		\begin{itemize}
		\item
			Suppose $s = f^k(t)$ and $t \in \Vars$ for some $k > 0$.
			Since $\WeightOf{s} = \WeightOf{t}$, $\WeightOf{f} = 0$.
			If $k = 1$, then we are done by case \prettyref{item:WPO-simp}.
			Otherwise $f^{k-1}(t) \GT_\KBO t$ by case
			\prettyref{item:KBO-simp} of \prettyref{def:KBO}.
			By the induction hypothesis we get $f^{k-1}(t) \GT_\WPOsum t$, and hence
			case \prettyref{item:WPO-simp} of \prettyref{def:WPO} applies.
		\item
			Suppose $t = g(\Seq{t_m})$ and case \prettyref{item:KBO-prec} or
			\prettyref{item:KBO-mono} applies.
			For all $j \in \{ 1, \dots, m \}$,
			we have $t \GT_\KBO t_j$ by the subterm property of $\GT_\KBO$,
			and we get $s \GT_\KBO t_j$ by the transitivity.
			By the induction hypothesis, $s \GT_\WPOsum t_j$.
			Hence, the side condition in \prettyref{item:WPO-args} of
			\prettyref{def:WPO} is satisfied, and subcase
			\prettyref{item:WPO-prec} or \prettyref{item:WPO-mono} applies.
		\end{itemize}
	\item
		Suppose $s \GT_\WPOsum t$.
		If $s >_\Asum t$, then $\WeightOf{s} > \WeightOf{t}$ by \prettyref{lem:weight} and
		$s \GT_\KBO t$ by \prettyref{item:KBO-gt} of \prettyref{def:KBO}.
		Otherwise we get $\WeightOf{s} = \WeightOf{t}$ by \prettyref{lem:weight}.
		\begin{itemize}
		\item
			Suppose $s_i \GE_\WPOsum t$ for some $i \in \{1,\dots,n\}$.
			By the induction hypothesis, we have $s_i \GE_\KBO t$.
			The subterm property of $\GT_\KBO$ ensures $s \GT_\KBO s_i$.
			Hence by \REV{}{the }transitivity, we get $s \GT_\KBO t$.
		\item
			Suppose $t = g(\Seq{t_m})$.
			If $f \PGT g$, then case \prettyref{item:KBO-prec} of \prettyref{def:KBO} applies.
			If $f = g$ and $[\Seq{s_n}] \GT_\WPOsum^\Lex [\Seq{t_m}]$, 
			then by the induction hypothesis we get
			$[\Seq{s_n}] \GT_\KBO^\Lex [\Seq{t_m}]$,
			and hence case \prettyref{item:KBO-mono} applies.
			\qedhere
		\end{itemize}
	\end{itemize}
\end{proof}

Note that we need neither admissibility nor $\Wzero > 0$ in \prettyref{cor:WPOsum}.
Let us see that removal of these conditions \REV{is}{are} indeed advantageous.
The following example illustrates that $\WPOsumP$ properly
\REV{extends}{enhances} KBO because \REV{}{the }admissibility is relaxed.

\begin{example}\label{ex:WPO>KBO+LPO}
	Consider the following TRS $\RR_1$:
	\[
		\RR_1 \DefEq
		\begin{EqSet}
			\f(\g(x))&\to\g(\f(\f(x)))\\
			\f(\h(x))&\to\h(\h(\f(x)))\\
		\end{EqSet}
	\]
	The first rule
\REV{%
	is oriented from right to left
}{%
	cannot be oriented
}%
	by $\LPO$ in any precedence.
	The second rule
\REV{%
	is oriented from right to left by KBO, since
	$\WeightOf{\h} > 0$ implies increase in weights and
	$\WeightOf{\h} = 0$ implies $\h \PGS \f$ by the admissibility.
}{%
	cannot be oriented by $\KBO$,
	since it requires that $\f \PGT \h$ and $\WeightOf{\h} = 0$
	which is not admissible.
}%
	On the other hand, $\WPOsumP$ with precedence $\f \PGT \g$, $\f \PGT \h$ and
	$\WeightOf{\g} > \WeightOf{\f} = \WeightOf{\h} = 0$ orients all the rules.
	Hence, $\RR_1$ is orientable by $\WPOsumP$, but not by $\KBO$ or $\LPO$.
	Note that there is no need to consider a status for $\RR_1$,
	since all symbols are unary.
\end{example}

Moreover, allowing $\Wzero = 0$ is also a proper enhancement.

\begin{example}\label{ex:WPOsum}
	Consider the following TRS $\RR_2$:
	\[
		\RR_2 \DefEq
		\begin{EqSet}
			\f(\ca,\cb) &\to \f(\cb,\f(\cb,\ca))\\
			\f(\ca,\f(\cb,x)) &\to \f(x,\f(\cb,\cb))
		\end{EqSet}
	\]
	The first rule cannot be oriented by $\KBO$ or $\WPOsumP$, since
	$\WeightOf{\cb} = 0$ is required.
	The second rule is not orientable by $\LPO$
	no matter \REV{how one chooses}{the choice of} $\sigma$.
	On the other hand, $\WPOsum$ with
	$\WeightOf{\ca} > \WeightOf{\cb} = \WeightOf{\f} = 0$,
	$\ca \PGT \cb$ and $\sigma(f) = [1,2]$ orients \REV{}{the }both rules.
	Hence, $\RR_2$ is orientable by $\WPOsum$ with $\Wzero = 0$,
	but not by $\LPO$, $\KBO$, or $\WPOsumP$.
\end{example}

\subsection{WPO$(\Apol)$}

Next we consider generalizing $\WPOsum$ using monotone polynomial interpretations.
According to Zantema \cite[Proposition 4]{Z01},
every monotone interpretation on a totally ordered set is weakly simple.
Hence a monotone polynomial interpretation
$\Apol$ is weakly simple and we obtain the following:

\begin{corollary}
	If $\Apol$ is strictly monotone, then
$\GT_\WPOpol$ is a reduction order.\qed
\end{corollary}

Trivially, POLO is subsumed by $\WPOpol$ as a reduction order.
More precisely, the following relation holds:

\begin{theorem}\label{thm:WPO>=POLO}
	${>_\Apol} \subseteq {\GT_\WPOpol}$.\qed
\end{theorem}

In the remainder of this paper, we consider 
\REV{an algebra $\Apol$ that}{$\Apol$}
consists of linear polynomial interpretations induced by 
a weight function $\Tp{\Weight,\Wzero}$ and a subterm coefficient function $\Coef$,
which is defined as follows:
\[
	f_\Apol(\Seq{a_n}) \DefEq
	\WeightOf{f} + \sum_{i=1}^{n} \CoefOf{f,i} \cdot a_i
\]
Analogous to \prettyref{thm:WPO>=KBO}, we also obtain the following:

\begin{theorem}\label{thm:WPO>=TKBO}
	If $\Wzero > 0$ and $\Weight$ is admissible for $\PGS$, then
	${\GT_\WPOpol} = {\GT_\TKBO}$.\qed
\end{theorem}

Moreover, we can verify that $\WPOpol$ strictly enhances both POLO and TKBO.
\REV{%
	More precisely, we show that both POLO and TKBO do not subsume even 
	WPO($\Asum$).
}{%
}%
\begin{example}\label{ex:POLO+TKBO<WPO}
	POLO cannot orient the first rule of $\RR_1$:
	\[
		l_1 = \f(\g(x)) \to \g(\f(\f(x))) = r_1
	\]
	since it is not $\omega$-terminating \cite{Z94}.
	Suppose that $\RR_1$ is oriented by TKBO. For the first rule, we need
	\[
		\VCoefOf{x,l_1} = \CoefOf{\f,1}\cdot\CoefOf{\g,1} \ge
		\CoefOf{\g,1}\cdot\CoefOf{\f,1}^2 = \VCoefOf{x,r_1}
	\]
	Hence $\CoefOf{\f,1} = 1$.
	Moreover,
	\begin{align*}
		\WeightOf{l_1} 
		&= \WeightOf{\f} + \WeightOf{\g} + \CoefOf{\g,1} \cdot \Wzero
	\\
		&\ge
		\WeightOf{\g} + \CoefOf{\g,1} \cdot (2 \cdot \WeightOf{\f} + \Wzero) =
		\WeightOf{r_1}
	\end{align*}
	Hence $\WeightOf{\f} = 0$.
	Analogously, for the second rule of $\RR_1$:
	\[
		l_2 = \f(\h(x)) \to \h(\h(\f(x))) = r_2
	\]
	we need 
	$\CoefOf{\h,1} = 1$ and $\WeightOf{\h} = 0$.
	Hence $\WeightOf{l_2} = \WeightOf{r_2}$.
\REV{%
	The admissibility imposes $\f \PSIM \h$,
	and thus the rule is oriented only from right to left.

	Note also that it is not possible to
	orient one of the rules in $\RR_1$ by $>_\Apol$ and
	the other by $\ge_\Apol$.
	Thus, togather with the discussion in \prettyref{ex:WPO>KBO+LPO},
	we conclude that 
	$\RR_1$ cannot be oriented by any lexicographic composition of
	POLO, (T)KBO, and LPO.
}{%
	By the admissibility,
	$\f \PGT \h$ cannot hold and this rule
	cannot be oriented by TKBO.
}%
\end{example}

\subsection{WPO$(\Amax)$}

Note that $\Apol$ is strictly monotone.
$\WPO$ also admits \emph{weakly} monotone interpretations;
a typical example is $\max$.%
\footnote{
\REV{%
	Note that weakly monotone polynomials with $0$ coefficients are not weakly simple,
	and hence cannot be applied for WPO of this section.
	We consider such polynomials in \prettyref{sec:pair}.
}{}%
}
Let us consider an instance of $\WPO$ using $\max$ for interpretation.

\begin{definition}
	A \emph{subterm penalty function} $\Pen$ is a mapping \st
	$\PenOf{f,i} \in \Nat$ is defined for each $f \in \Sig_n$ and
	$i \in \{ 1, \dots, n \}$.
	A weight function $\Tp{\Weight,\Wzero}$ and $\Pen$ induce
	the $\Sig$-algebra $\Amax$, which
	consists of the carrier set $\{ a \in \Nat \mid a \ge \Wzero \}$ and
	interpretations given by:
	\[
		f_\Amax(\Seq{a_n}) \DefEq
		\max
		\Big(\WeightOf{f}, \max_{i=1}^n \big(\PenOf{f,i} + a_i\big)
		\Big)
	\]
\end{definition}

\begin{lemma}
	$\Amax$ is weakly monotone and weakly simple.
\end{lemma}
\begin{proof}
	Weak simplicity is obvious from the fact that $\max( \dots, a, \dots ) \ge a$.
	For weak monotonicity, suppose $a > b$ and let us show
	\[
		a' = f_\Amax(\Seq{c_{k}}, a, \Seq{d_l})\ge
		f_\Amax(\Seq{c_{k}}, b, \Seq{d_l}) = b'
	\]
	To this end, let $c = f_\Amax(\Seq{c_{k}}, 0, \Seq{d_l} )$.
	If $c \ge \PenOf{f,k+1} + a$, then $a' = b' = c$.
	Otherwise, we have $a' = \PenOf{f,k+1} + a$ and either
	$a' > \PenOf{f,k+1} + b = b' > c$ or $a' > b' = c$.\qedhere
\end{proof}
Note that $\Amax$ can be considered as the \REV{1-dimensional}{dimension-1} variant of
\emph{arctic interpretations} \cite{KW09}.
The weak monotonicity of $\Amax$ is also shown there.

\begin{corollary}
	$\GT_\WPOmax$ is a reduction order.\qed
\end{corollary}

Now we show that LPO is obtained as a restricted case of $\WPOmax$.

\begin{theorem}\label{thm:WPO>=LPO}
	If $\Wzero = 0$, $\WeightOf{f} = 0$ and
	$\PenOf{f,i} = 0$ for all $f \in \Sig_n$ and $i \in \{ 1, \dots, n \}$, then
	${\GT_\LPO} = {\GT_\WPOmax}$.
\end{theorem}
\begin{proof}
	From the assumptions, $s >_\Amax t$ never holds.
	Hence, case \prettyref{item:WPO-gt} of \prettyref{def:WPO} can be ignored.
	Moreover,
	$s \ge_\Amax t$ is equivalent to $\Var(s) \supseteq \Var(t)$.
	One can easily verify the latter holds whenever $s \GT_\LPO t$,
	using the fact that $s \NGE_\LPO x$ for $x \notin \Var(s)$.
	Hence, the condition of 
	\REV{case }{}\prettyref{item:WPO-ge} can be ignored and
	\prettyref{def:LPO} and \prettyref{def:WPO} become equivalent.\qedhere
\end{proof}

The following example illustrates that $\WPOmax$ properly enhances $\LPO$.

\begin{example}
	Consider the following TRS $\RR_3$:
	\[
		\RR_3 \DefEq
		\begin{EqSet}
			\f(x,y) &\to \g(x)\\
			\f(\g(x),y)&\to \f(x,\g(x))\\
			\f(x,\g(y)) &\to \f(y,y)
		\end{EqSet}
	\]
	To orient the first two rules by $\LPO$,
	we need $\f \PGT \g$ and $\sigma(\f) = [1,2]$.
	$\LPO$ cannot orient the third rule by this precedence and status,
	while $\PenOf{\g,1} > \PenOf{\f,1} = 0$ suffices for $\WPOmax$.
	Since the last two rules are duplicating,
	$\KBO$ or $\WPOsum$ cannot apply for $\RR_3$.
\end{example}

However, $\WPOmax$
\REV{%
	covers neither $\WPOsum$ nor even KBO.
}{%
	does not cover $\WPOsum$, \REV{and}{or} not even $\KBO$.
}%
In the next section, we consider unifying $\WPOsum$ and $\WPOmax$
to cover both KBO and LPO.

\subsection{WPO$(\Amp)$ and WPO$(\Ams)$}

Now we consider unifying $\WPOmax$ and $\WPOpol$.
To this end, we introduce the \emph{weight status} to choose a polynomial or $\max$
for each function symbol.

\begin{definition}
	A \emph{weight status function} is a mapping $\Wstatus$ which maps
	each function symbol $f$
\REV{%
	either to the symbol $\Pol$ or to the symbol $\Max$.
}{%
	to either symbol $\Pol$ or $\Max$.
}%
	The $\Sig$-algebra $\Amp$ consists of the carrier set 
	$\{ a \in \Nat \mid a \ge \Wzero \}$ and the interpretation which is defined as follows:
	\[
		f_\Amp(\Seq{a_n}) \DefEq
		\begin{Cases}
			\WeightOf{f} + \displaystyle\sum_{i = 1}^{n} \CoefOf{f,i} \cdot a_i
			&\text{if } \WstatusOf{f} = \Pol
		\smallskip\\
			\max
			\Big( \WeightOf{f}, \displaystyle\max_{i = 1}^{n} \big(\PenOf{f,i} + a_i\big)
			\Big)
			&\text{if } \WstatusOf{f} = \Max
		\end{Cases}
	\]
	We denote $\Amp$ by $\Ams$ if coefficients are at most $1$.
\end{definition}

\begin{corollary}
	$\WPOmp$ is a reduction order.\qed
\end{corollary}

Trivially, $\WPOms$ \REV{subsumes}{encompasses} both $\WPOsum$ and $\WPOmax$.
Hence, we obtain the following more \REV{interesting}{influential} result:

\begin{theorem}\label{thm:WPO>=LPO+KBO}
	$\WPOms$ \REV{subsumes}{encompasses} both LPO and KBO.\qed
\end{theorem}

As far as we know, this is the first reduction order
that \REV{unifies}{unify} LPO and KBO.
The following example illustrates that
$\WPOms$ is strictly stronger than the union of $\WPOsum$ and $\WPOmax$.

\begin{example}
	Consider the following TRS $\RR_4$:
	\[
		\RR_4 \DefEq
		\begin{EqSet}
			\f(\f(x,y),z) &\to \f(x,\f(y,z))\\
			\g(\f(\ca,x),\cb) &\to \g(\f(x,\cb),x)
		\end{EqSet}
	\]
	If $\WstatusOf{\f} = \Max$, then the first rule requires $\PenOf{\f,2} = 0$.
	Under this restriction the second rule cannot be oriented.
	If $\WstatusOf{\f} = \Pol$, then the first rule is
	oriented iff $\CoefOf{f,1} = \CoefOf{f,2} = 1$ and $\sigma(f) = [1,2]$.
	On the other hand, the duplicating variable $x$ in the second rule requires
	$\WstatusOf{\g} = \Max$.
	Hence, $\RR_4$ is orientable by $\WPOms$ only if $\WstatusOf{\f} = \Pol$ and
	$\WstatusOf{\g} = \Max$.
\end{example}

\REV{%
	The following example
}{%
	Let us close this section with an example that
}%
suggests \REV{that }{}$\WPOms$
advances the state-of-the-art of automated termination proving.

\begin{example}
	The most powerful termination provers including
	\AProVEver{2013} and \TTTTver{1.11}
	fail to prove termination of the following TRS $\RR_5$:
	\[
		\RR_5 \DefEq
		\begin{EqSet}
			\f(\g(\g(x,\ca),\g(\cb,y))) &\to \f(\g(\g(\h(x,x),\cb),\g(y,\ca)))
		\\
			\g(x,y) &\to x
		\\
			\h(x,\h(y,z)) &\to y
		\end{EqSet}
	\]

	We show that $\WPOms$ with
	$\WstatusOf{\g} = \Pol$, $\WstatusOf{\h} = \Max$,
	$\WeightOf{\ca} > \WeightOf{\cb}$,
	$\WeightOf{\h} = \PenOf{\h,1} = \PenOf{\h,2} = 0$
	and $\sigma(\f) = \sigma(\g) = [1,2]$
	orients all the rules.
	For the first rule, applying case \prettyref{item:WPO-mono} twice
	it yields orienting $\g(x,\ca) \GT_\WPOms \g(\h(x,x),\cb)$ where
	case \prettyref{item:WPO-gt} applies.
	The other rules are trivially oriented.
\end{example}

\REV{%
\newcommand\mat[1]{
	\begin{pmatrix}
		#1
	\end{pmatrix}
}%
	In general,
	matrix interpretations cannot be combined with WPO,
	since a matrix interpretation is often not weakly simple.
	Consider the following TRS from \cite{EWZ08}:
	\[
		\RR_6 \DefEq \left\{ \f(\f(x)) \to \f(\g(\f(x)))\right\}
	\]
	which is shown terminating by the following matrix interpretation $\Amat$ \st
	\begin{align*}
		\f_\Amat(\vec{x}) &= \mat{1&1\\0&0} \cdot \vec{x} + \mat{0\\1}&
		\g_\Amat(\vec{x}) &= \mat{1&0\\0&0} \cdot \vec{x}
	\end{align*}
	However, $\g_\Amat$ is not weakly simple. For example,
	\[
		\g_\Amat(\mat{0\\1}) = \mat{1&0\\0&0} \cdot \mat{0\\1} = \mat{0\\0}
		\ngeq \mat{0\\1}
	\]
	Hence to unify the matrix interpretation with WPO,
	we have to further relax the weak simplicity condition.
	This is achieved in the next section by extending WPO to a reduction pair.
}{}%

\section{WPO as a Reduction Pair}\label{sec:pair}

In this section, we extend WPO to a reduction pair.
First we introduce the basic definition and prove its soundness,
and then \REV{we }{}present two refinements.
\REV{Afterwards}{Then} we introduce some instances of WPO as reduction pairs
and investigate relationships with existing reduction pairs.
In particular, matrix interpretations are
also subsumed by WPO as a reduction pair.

\subsection{WPO with Partial Status}\label{sec:pair definition}

In the preliminary version of this paper \cite{YKS13},
we simply applied argument filtering to obtain the reduction pair
$\Tp{{\GS_\WPO^\pi},{\GT_\WPO^\pi}}$ from WPO.
In this paper, we fully revise this approach and directly define a
reduction pair by incorporating
\emph{partial statuses} \cite{YKS13b} into WPO.
A partial status is a generalization of status
that admits \emph{non-permutations}.

\begin{definition}
	A \Def{partial status function} $\sigma$ is a mapping
	that assigns \REV{to }{}each $n$-ary symbol $f$ a list $\ListOf{\Seq{i_m}}$
	of (distinct) positions in $\{ 1, \dots, n \}$.
	We also view $\sigma(f)$ as the set $\SetOf{\Seq{i_m}}$
	for $\sigma(f) = \ListOf{\Seq{i_m}}$.
	A well-founded algebra $\A$ is \Def{(weakly) simple}
	\wrt $\sigma$ iff
	$f_\A$ is (weakly) simple in \REV{its }{}$i$-th argument
	for every $f \in \Sig$ and $i \in \sigma(f)$.
\end{definition}

Note that here $\sigma(f)$ need not be a permutation,
\REV{since}{and}
some positions may be ignored.
If every $\sigma(f)$ is a permutation,
then we say \REV{that }{}$\sigma$ is \Def{total}.
Conversely if $\sigma(f) = \List\Empty$ for every $f$, then
we call $\sigma$ the \Def{empty} status.

\begin{definition}[WPO with Partial Status]\label{def:WPOpS}
	Let $\A$ be a well-founded algebra and $\sigma$ a partial status.
	The pair $\Tp{\GS_\WPOAS,\GT_\WPOAS}$ of relations 
	is defined mutually recursively as follows:
	$x \GS_\WPOAS x$, and
	$s = f(\Seq{s_n}) \GSopt_\WPOAS t$ iff
	\begin{enumerate}
	\item\label{item:WPO-gt}
		$s \AGT t$, or
	\item\label{item:WPO-ge}
		$s \AGS t$ and
		\begin{enumerate}
		\item\label{item:WPO-simp}
			$\ForSome{i \in \sigma(f)} s_i \GS_\WPOAS t$, or
		\item\label{item:WPO-args}
			$t = g(\Seq{t_m})$,
			$\ForAll{j \in \sigma(g)} s \GT_\WPOAS t_j$ and either
			\begin{enumerate}
			\item\label{item:WPO-prec}
				$f\PGT g$ or
			\item\label{item:WPO-mono}
				$f\PSIM g$ and
				\(
					\AppPerm{\sigma(f)}{s}{n} \GSopt_\WPOAS^\Lex
					\AppPerm{\sigma(g)}{t}{m}
				\).
			\end{enumerate}
		\end{enumerate}
	\end{enumerate}
\end{definition}

\REV{%
	In the appendix we prove that 
	the pair $\Tp{\GS_\WPO,\GT_\WPO}$ is indeed a reduction pair.

}{}%
The effect of a partial status has similarity with that of
combining argument filtering and a standard total status.
Indeed, WPO with partial status subsumes
WPO with total status and
\REV{certain form of}{\emph{non-collapsing}} argument filtering.
\REV{%
	An argument filter $\pi$ is said to be \Def{non-collapsing} iff
	$\pi(f)$ is a list for every $f \in \Sig$.
}{}%

\begin{proposition}\label{prop:partial status >= argument filter}
	Let $\pi$ be a non-collapsing argument filter\REV{.}{, \REV{\ie,}{\ie}
	$\pi(f)$ is a list for every $f \in \Sig$.}
	For every $\Sig^\pi$-algebra $\A$ and total status $\sigma$ on $\Sig^\pi$,
	there exists an $\Sig$-algebra $\A'$ and a partial status $\sigma'$ on $\Sig$ \st
	\[
		\Tp{{\GS_\WPOAS^\pi},{\GT_\WPOAS^\pi}} =
		\Tp{{\GS_{\WPO(\A',\sigma')}},{\GT_{\WPO(\A',\sigma')}}}
	\]
\end{proposition}
\begin{proof}
	Let us define the interpretation of each $f \in \Sig_n$ in $\A'$ by
	$f_{\A'}(\Seq{x_n}) \DefEq f_\A(\Seq{x_{\pi(f)}})$.
	Then obviously, $\pi(s) \gsopt_\A \pi(t)$ iff $s \gsopt_{\A'} t$.
	Moreover, we define $\sigma'(f) $ by $\pi(f) \star \sigma(f)$,
	where $\star$ is a left-associative operator defined by
	\[
		[a_1,\dots,a_n] \star [i_1,\dots,i_{n'}] \DefEq [a_{i_1},\dots,a_{i_{n'}}]
	\]

	Now we verify that 
	$s \GSopt_\WPOAS^\pi t$ implies $s \GSopt_{\WPO(\A',\sigma')} t$
	by induction on $|s| + |t|$.
	If $s \in \Vars$, then
	$s = \pi(t) = t$ since $\pi$ is non-collapsing, and hence
	$s \GS_{\WPO(\A',\sigma')} t$.
	Suppose $s = f(\Seq{s_n})$.
	We proceed \REV{by}{to} case
	\REV{analysis on the}{splitting for} derivation of $\pi(s) \GSopt_\WPOAS \pi(t)$.
	\begin{enumerate}
	\item
		Suppose $\pi(s) \AGT \pi(t)$.
		We obviously have $s >_{\A'} t$ and hence $s \GT_{\WPO(\A',\sigma')} t$.
	\item
		Suppose $\pi(s) \AGS \pi(t)$.
		We obviously have $s \gs_{\A'} t$.
		\begin{enumerate}
		\item
			Suppose that $\pi(s) \GSopt_\WPOAS \pi(t)$ is derived by case \prettyref{item:WPO-simp}.
			Then we have $s_i \GS_\WPOAS^\pi t$ 
			for some $i \in [1,\dots,n] \star \pi(f)$.
			By the induction hypothesis we have $s_i \GS_{\WPO(\A',\sigma')} t$,
			and since $i \in \sigma'(f)$,
			$s \GT_{\WPO(A',\sigma')} t$ by case \prettyref{item:WPO-simp}.
		\item
			Suppose that $\pi(s) \GSopt_\WPOAS \pi(t)$ is derived by case \prettyref{item:WPO-args}.
			Then we have $t = g(\Seq{t_m})$ and $s \GT_\WPOAS^\pi t_j$
			for every $j \in [1,\dots,m] \star \pi(g)$.
			By the induction hypothesis we have $s \GT_{\WPO(\A',\sigma')} t_j$,
			for all $j \in \sigma'(f)$.
			If furthermore $f \PGT g$, then immediately $s \GT_{\WPO(\A',\sigma')} t$
			by case \prettyref{item:WPO-prec}.
			If case \prettyref{item:WPO-mono} applies,
			then we obtain
			\[
				[\Seq{s_n}] \star \pi(f) \star \sigma(f) \GSopt_\WPOAS^{\pi\;\Lex}
				[\Seq{t_m}] \star \pi(g) \star \sigma(g)
			\] 
			By the induction hypothesis and definition of $\sigma'$ we obtain
			\[
				\AppPerm{\sigma'(f)}{s}{n} \GSopt_{\WPO(\A',\sigma')} \AppPerm{\sigma'(g)}{t}{m}
				\tag*{\qed}
			\]
		\end{enumerate}
	\end{enumerate}
\end{proof}

The advantage of partial status over argument filtering
is due to the \emph{weights} of ignored arguments.
This is illustrated by the following example.

\begin{example}\label{ex:predecessor}\cite{YKS13b}
	Consider
	a DP problem that induces the following constrains:
	\begin{align*}
		\F(\s(x)) &\GT \F(\p(\s(x)))&
		\p(\s(x)) &\GS x
	\end{align*}
	In order to
\REV{satisfy the first constraint}{strictly orient the rule in $\PP$}
	by any simplification order,
	the argument of $\p$ must be filtered.
	However,
\REV{the second constraint cannot be satisfied}{the rule in $\RR$ cannot be weakly oriented}
	under such an argument filtering.

	On the other hand, the DP problem can be shown finite
	using $\WPO$ with partial status \st
	$\s_\A(x) > x$, $\F_\A(x) = \p_\A(x) = x$,
	$\sigma(\F) = \sigma(\s) = [1]$, $\sigma(\p) = \List\Empty$, and
	$\s \PGT \p$.
	We have
	$\F(\s(x)) \GT_\WPO \F(\p(\s(x)))$ because of cases
	\prettyref{item:WPO-mono} and \prettyref{item:WPO-prec}, and
	$\p(\s(x)) \GS_\WPO x$ because of case \prettyref{item:WPO-gt}.
\end{example}

\subsection{Refinements}\label{sec:pair refinements}

As we will see in \prettyref{sec:experiments pair},
the reduction pair processor induced by \prettyref{def:WPOpS} is
already powerful in practice.
However in theory, 
\REV{the }{}WPO reduction pair is not a proper extension of the
underlying interpretation, \eg,
$x \GS_\A g(x)$ if $g_\A(x) = x$,
but $x \GS_\WPO g(x)$ cannot hold.
Hence we refine the definition of $\GS_\WPO$ to properly subsume
the underlying interpretation.

\REV{%
	Note that in the above case,
	assuming $x \GS_\WPO g(x)$ does not cause a problem
	if $g$ is \emph{least} (\ie, $f \PGS g$ for every $f \in \Sig$) and
	$\sigma(g) = \List\Empty$.
}{}%
\begin{proposition}\label{prop:least}
	Let $g \in \Sig$ \st $f \PGS g$ for every $f \in \Sig$ and $\sigma(g) = \List\Empty$.
	Then $x \AGS t = g(\Seq{t_m})$ implies
	$s \GS_\WPO t\Subst{x \TO s}$ for arbitrary non-variable
	\REV{terms }{}$s = f(\Seq{s_n})$.
\end{proposition}
\begin{proof}
	By the definition of $\AGS$, we have $s \AGS t\Subst{x \TO s}$.
	Since $g$ is \REV{least}{minimal} \wrt $\PGS$, we have $f \PGS g$.
	Moreover, since $\sigma(g) = \List\Empty$, we have 
	$\AppPerm{\sigma(f)}{s}{n} \GS_\WPO^\Lex \AppPerm{\sigma(g)}{t}{m} = \List\Empty$.
	\qedhere
\end{proof}

\prettyref{prop:least} suggests a refined definition of
$s \GS_\WPO t$ by adding
the following subcase in case \prettyref{item:WPO-ge} of \prettyref{def:WPOpS}
(note that $s \AGS t$ is ensured in this case):
\begin{enumerate}
\setcounter{enumi}{2}
\item[]
\begin{enumerate}[(2a)]
\setcounter{enumii}{2}
\item\label{item:WPO-min}
	$s \in \Vars$ and $t = g(\Seq{t_m})$ \st
	$\sigma(g) = \List\Empty$ and $g$ is \REV{least}{minimal} \wrt $\PGS$.
\end{enumerate}
\end{enumerate}

Similar refinements are proposed for KBO
\REV{\cite{KV03,ST13}}{\cite{KV03}} and for
RPO \cite{TAN12}, when $t$ is a \REV{least}{minimal} constant.
Our version is more general since $t$ need not be a constant.

\begin{example}\cite{YKS13b}
	Consider a DP problem that induces the following constraints:
	\begin{align*}
		\F(\s(x),y) &\GT \F(\p(\s(x)),\p(y))
	\\
		\F(x,\s(y)) &\GS \F(\p(x),\p(\s(y)))
	\\
		\p(\s(x))	&\GS x
	\end{align*}
	Let
	$\sigma(\p) = \List\Empty$, $\sigma(\F) = [1]$,
	$\s_\A(x) = x + 1$, $\p_\A(x) = x$, $\F_\A(x,y) = x$ and
	$\p$ be \REV{least}{minimal} \wrt $\PGS$.
	The first constraint is strictly oriented by case \prettyref{item:WPO-prec}.
	For the second constraint, it yields $x \GS_\WPO \p(x)$,
	for which case (2c) applies.
	Note that the argument of $\p$ cannot be filtered by an argument filter,
	because of the third constraint.
	Hence the refinements of \REV{\cite{KV03,ST13}}{\cite{KV03}} or \cite{TAN12} do not work for this example.
\end{example}

Moreover, a further refinement is also possible for WPO,
when the right\REV{-}{ }hand side is a variable.
\REV{%
	Note that $f(x) \AGS x$ does not imply $f(x) \GS_\WPO x$ if $\sigma(f) = \List\Empty$.
	Nonetheless,
	$f(x) \GS_\WPO x$ can be assumed if
	$f$ is greatest, and
	moreover $f \PSIM g$ implies $\sigma(g) = \List\Empty$.
	The latter condition is crucial,
	since $g(x) \GT_\WPO f(g(x))$ if $f \PSIM g$ and $g = [1]$.
}{}%

\begin{proposition}\label{prop:greatest}
	Suppose that $\A$ is strictly simple \wrt $\sigma$, and
	$f \in \Sig$ \st
	either $f \PGT g$\REV{,}{} or $f \PSIM g$ and $\sigma(g) = \List\Empty$ for every $g \in \Sig$.
	Then $s = f(\Seq{s_n}) \AGS y$ implies
	$s\Subst{y \TO t} \GS_\WPO t$ for arbitrary non-variable
	\REV{term }{}$t = g(\Seq{t_m})$.
\end{proposition}
\begin{proof}
	By the definition of $\AGS$, we have $s\Subst{y \TO t} \AGS t$.
	Moreover by the strict simplicity,
	$s\Subst{y \TO t} \AGS t \AGT t_j$ for all $j \in \sigma(g)$.
	Hence we get $s \GT_\WPO t_j$.
	If $f \PGT g$, then $s\Subst{y \TO t} \GT_\WPO t$ by case \prettyref{item:WPO-prec}.
	If $f \PSIM g$ and $\sigma(g) = \List\Empty$, then
	$s \Subst{y \TO t} \GS_\WPO t$ by case \prettyref{item:WPO-mono}.
	\qedhere
\end{proof}

Provided $\A$ is strictly simple \wrt $\sigma$,
\prettyref{prop:greatest} suggests a refinement of $s \GS_\WPO t$
by adding the following subcase in case \prettyref{item:WPO-ge}:
\begin{enumerate}
\setcounter{enumi}{2}
\item[]
\begin{enumerate}[(2a)]
\setcounter{enumii}{3}
\item\label{item:WPO-max}
	$s = f(\Seq{s_n})$ and $t \in \Vars$ \st for every $g \in \Sig$,
	either $f \PGT g$ or $f \PGS g$ and $\sigma(g) = \List\Empty$.
\end{enumerate}
\end{enumerate}

Note that an arbitrary algebra is strictly simple \wrt the empty status.
Hence
\REV{%
	WPO with the refinements can subsume even matrix interpretations. We
}{%
	we
}%
obtain the following result:

\begin{theorem}\label{thm:WPO refine}
\REV{%
	Consider an instance of WPO that is induced by
	\begin{itemize}
	\item
		a well-founded algebra $\A$ that is \emph{non-trivial},
		\ie, there exist $a,b \in A$ \st $a \ngs b$,
	\item
		the empty status function $\sigma$, and
	\item
		the quasi-precedence $\PGS$ \st
		$f \PGS g$ for arbitrary $f, g \in \Sig$.
	\end{itemize}\noindent
}{%
	Let $\A$ be a non-trivial well-founded algebra,
	$\sigma$ the empty status function and $\PGS$ the quasi\REV{-}{ }precedence
	$\Sig^2$.
}%
	Then, $\Tp{{\AGS},{\AGT}} = \Tp{{\GS_\WPO},{\GT_\WPO}}$ after the refinements.
\end{theorem}
\begin{proof}
	From the definition, it is obvious that ${\AGT} \subseteq {\GT_\WPO}$ and
	${\GS_\WPO} \subseteq {\AGS}$.
	By the assumptions,
	cases \prettyref{item:WPO-prec} and \prettyref{item:WPO-mono} of \prettyref{def:WPOpS}
	cannot apply for $\GT_\WPO$. Hence we easily obtain ${\GT_\WPO} \subseteq {\AGT}$.

	Now suppose $s \gs_\A t$ and let us show $s \GS_\WPO t$.
	The proof proceeds \REV{by}{to} case \REV{analysis on}{splitting of}
	the structure of $s$ and $t$.
	\begin{itemize}
	\item
		If $s, t \in \Vars$, then from non-triviality
		we have $s = t$, and hence $s \GS_\WPO t$.
	\item
		Suppose $s = f(\Seq{s_n})$ and $t = f(\Seq{t_m})$.
		Since $f \PGT g$ never hold\REV{s}{},
		case \prettyref{item:WPO-prec} of \prettyref{def:WPOpS} can be ignored.
		Moreover, since $f \PGS g$ and $\List\Empty \GS_\WPO^\Lex \List\Empty$,
		$s \AGS t$ implies $s \GS_\WPO t$.
	\item
		If either $s$ or $t$ is a variable, then
		\REV{}{either }refinement \prettyref{item:WPO-min} or \prettyref{item:WPO-max}
		is satisfied. Hence $s \GS_\WPO t$.
		\qedhere
	\end{itemize}
\end{proof}

\REV{%
	In the appendix, we prove soundness of WPO with the refinements
	\prettyref{item:WPO-min} and \prettyref{item:WPO-max}.

	\begin{theorem}[Soundness]\label{thm:WPO pair}
		If $\A$ is weakly monotone and
		weakly simple \wrt $\sigma$, then
		$\Tp{\GS_\WPO,\GT_\WPO}$ forms a reduction pair.
	\end{theorem}\unskip
}{%
}%
\subsection{Comparison with Other Reduction Pairs}
\label{sec:pair instances}

In this section,
we investigate some relationships between
instances of WPO and existing reduction pairs.

In \prettyref{def:WPOpS}, it is obvious that
\REV{the }{}induced $\GT_\WPO$ is identical to that induced by \prettyref{def:WPO},
if we choose a \emph{total} status $\sigma$.
Hence
Theorems \ref{thm:WPO>=KBO}, \ref{thm:WPO>=TKBO} and \ref{thm:WPO>=LPO}
imply that WPO \REV{subsumes}{encompasses} KBO, TKBO and LPO \resp also as a reduction pair.

On the other hand,
\prettyref{thm:WPO>=POLO} does not imply
that $\WPOpol$ subsumes POLO as a reduction pair,
since the ``weak-part'' $\ge_\Apol$
is not considered in the theorem.
Nonetheless, after the refinements in \prettyref{sec:pair refinements},
we obtain the following result from \prettyref{thm:WPO refine}:

\begin{corollary}\label{cor:WPO>=POLO pair}
	\REV{}{As a reduction pair,}%
	$\WPOpol$
\REV{%
	with the refinements \prettyref{item:WPO-min} and \prettyref{item:WPO-max}
	subsumes
}{%
	encompasses
}%
	POLO.\qed
\end{corollary}

It is now easy to obtain the following result:

\begin{corollary}
	$\WPOmp$
\REV{%
	with the refinements \prettyref{item:WPO-min} and \prettyref{item:WPO-max}
	subsumes
}{%
	encompasses
}%
	POLO, KBO, TKBO and LPO.\qed
\end{corollary}

Moreover, WPO also subsumes 
the \emph{matrix interpretation method} \cite{EWZ08}, when
weights are computed by a matrix interpretation $\Amat$.
Note that a matrix interpretation is not always weakly simple.
Hence as a \emph{reduction order},
\prettyref{def:WPO} cannot be applied for $\Amat$ in general.
The situation is relaxed for reduction pairs, and
from \prettyref{thm:WPO refine} we obtain the following:

\begin{corollary}\label{cor:WPO>=MAT pair}
\REV{%
}{%
	The reduction pair induced by 
}%
	$\WPOmat$ \REV{with the refinements \prettyref{item:WPO-min} and \prettyref{item:WPO-max} subsumes}{encompasses}
	the reduction pair induced by the matrix interpretation $\Amat$.\qed
\end{corollary}

Finally, we compare $\WPOmp$ and
\emph{RPOLO} of Bofill \etal \cite{BBRR13},
another approach of unifying LPO and POLO.
It turns out that RPOLO is incomparable with $\WPOmp$.
First we verify that $\WPOmp$ is not subsumed by $\RPOLO$;
more precisely, RPOLO does not subsume KBO.

\begin{example}\label{ex:WPO-RPOLO}
	Let us show that the constraint
	$\f(\g(x)) \GT \g(\f(\f(x)))$ cannot be
	satisfied by RPOLO.%
\footnote{
	To simplify the discussion, we do not consider \REV{the }{}possibility for
	\emph{argument filterings} or 
	\emph{usable rules} \cite{AG00} in the following examples.
\REV{%
	It is easy to exclude these techniques;
	for example, by adding the constraint $\g(x) \GS x$ we can enforce
	the argument of $\g$ not to be filtered.
}{%
	Nonetheless, it is easy to exclude these techniques by adding rules \eg $\g(x) \to x$.
}%
}
	Note that this constraint is satisfied
	by KBO with $\WeightOf{\f}=0$ and $\f \PGT \g$.
	\begin{itemize}
	\item
		Suppose $f \in \Sig_\POLO$. Since this
		constraint cannot be satisfied by POLO,
		$\g$ must be in $\Sig_\RPO$.
		Hence we need
		$\f_\Apol(v_{\g(x)}) >_{C(\f(\g(x)))} v_{\g(\f(\f(x)))}$.
		This requires either
		\begin{itemize}
		\item $\f_\Apol(x) = x$ and $\g(x) \GT_\RPOLO \g(\f(\f(x)))$, or
		\item $\f_\Apol(x) > x$ and $\g(x) \GE_\RPOLO \g(\f(\f(x)))$.
		\end{itemize}
		In either case, we obtain $\g(x) \GT_\RPOLO \g(x)$,
		which is a contradiction.

	\item
		Suppose $\f \in \Sig_\RPO$. Since this
		constraint cannot be satisfied by RPO,
		$\g$ must be in $\Sig_\POLO$.
		Hence we need
		\begin{itemize}
		\item $\g(x) \GE_\RPOLO \g(\f(\f(x)))$, or
		\item $\f(\g(x)) \GT_\RPOLO \f(\f(x))$.
		\end{itemize}
		The first case contradicts \REV{}{with }$\f(x) \GT_\RPOLO x$.
		The second case contradicts \REV{}{with }the fact that $\f(x) \GT_\RPOLO \g(x)$.
	\end{itemize}
\end{example}

On the other hand, $\RPOLO$ is also not subsumed by $\WPOmp$,
as the following example illustrates.

\begin{example}\label{ex:RPOLO-WPO}
	Consider a DP problem that induces the following constraints:
	\begin{align}
	\label{eq:g}
		\F(\ci(x,\ci(y,\g(z)))) &\GT \F(\ci(y,\ci(z,x))))
	\\\label{eq:Zantema}
		\f(\g(\h(x))) &\GS \g(\f(\h(\g(x))))
	\\\label{eq:i}
		\ci(y,\ci(z,x)) &\GS \ci(x,\ci(y,z))
	\end{align}
	where constraint \eqref{eq:Zantema} is from \cite[Proposition 10]{Z94}.
	\begin{itemize}
	\item
		First, let us show that
		the set of constraints cannot be satisfied
		by $\WPOmp$. Since $\f$, $\g$ and $\h$ are unary,
		we only consider $\Wstatus(\f) = \Wstatus(\g) = \Wstatus(\h) = \Pol$.
		It is easy to adjust \cite[Proposition 10]{Z94} to show that
		$\g_\Apol(x) = x$ whenever
		\eqref{eq:Zantema} is satisfied.
		Together with \eqref{eq:i}, we obtain
		\[
			\ci(y,\ci(z,x)) \ge_\Amp \ci(x,\ci(y,z)) =_\Amp \ci(x,\ci(y,\g(z)))
		\]
		Hence case \prettyref{item:WPO-gt} of
		\prettyref{def:WPOpS} cannot be applied
		for constraint \eqref{eq:g}.
		Moreover, by any choice of $\sigma(\ci)$,
		case \prettyref{item:WPO-mono} cannot apply, either.
	\item
		Second, let us show that the
		set of constraints can be satisfied by RPOLO.
		Consider $\f,\g,\h \in \Sig_\RPO$, $\ci,\F \in \Sig_\POLO$,
		$\f \PGT \g \PGT \h$,
		$\ci_\Apol(x,y) = x + y$ and $\F_\Apol(x) = x$.
		Then
		constraint \eqref{eq:Zantema} is strictly oriented
		and \eqref{eq:i} is weakly \REV{oriented}{so}.
		Since $\g(z) \GT_\RPOLO z$ and $v_{\g(z)} > z$ implies $x+y+v_{\g(z)} > y + z + x$,
		constraint \eqref{eq:g} is also satisfied.
		\qedhere
	\end{itemize}
\end{example}

\section{SMT Encodings}\label{sec:encodings}

In the preceding sections, we have concentrated on theoretical aspects.
In this section, we consider how to implement the instances of $\WPO$ using
SMT solvers.
We extend the corresponding approach for KBO \cite{ZHM09} to WPO.
In particular, $\WPOsum$, $\WPOmax$ and $\WPOms$ are reduced to 
SMT problems of linear arithmetic, and as a consequence, 
decidability is ensured for orientability problems of these orders.

An \emph{expression} $e$ is 
built from (non-negative integer) variables, constants and
the binary symbols $\cdot$ and $+$ denoting multiplication and addition, \resp.
A \emph{formula} is built from
atoms of the form $e_1 > e_2$ and $e_1 \ge e_2$,
negation $\Not$,
and the binary symbols $\And$, $\Or$ and $\Then$ denoting
conjunction, disjunction and implication, \resp.
The precedence of these symbols \REV{is}{are} in the order we listed above.

The main interest of the SMT encoding approach is to employ SMT solvers for
finding a concrete algebra that proves finiteness of a given DP problem
(or termination of a given TRS).
Hence we assume that algebras are
parameterized by a set of expression variables.

\begin{definition}
An algebra $\A$ \Def{parameterized} by a set $V$ of variables
is a mapping that induces a concrete algebra $\A^\alpha$ from an assignment $\alpha$
whose domain contains $V$.
An \Def{encoding} of the relation $\gsopt_\A$ is a function that assigns
for two terms $s$ and $t$ a formula
$\Encode{s \gsopt_\A t}$ over variables from $V$ \st
$\alpha \models \Encode{s \gsopt_\A t}$ iff $s \gsopt_{\A^\alpha} t$.
\end{definition}

In the encodings presented in the rest of this section,
we consider $\A$ \REV{to be}{is} parameterized by at most the following variables:
\begin{itemize}
\item
	integer variables $\WeiVarOf{f}$ and $\WeiVarZero$ denoting $\WeightOf{f}$ and $\Wzero$, \resp, and
\item
	integer variables $\CoefVarOf{f,i}$ and $\PenVarOf{f,i}$ 
	denoting $\CoefOf{f,i}$ and $\PenOf{f,i}$, \resp.
\end{itemize}

\subsection{The Common Structure}\label{sec:encode WPO}

\def\ST{\mathsf{ST}}
\def\SIMP{\mathsf{SIMP}}
To optimize the presentation,
we present an encoding of the common structure of $\WPO$
independent from the shape of $\A$.
Hence, we assume encodings for
$\AGT$ and $\AGS$ are given.

\REV{Following \cite{ZHM09}, first}{First}
we represent a \REV{quasi-}{}precedence $\PGS$ by integer variables $\PrecVarOf{f}$.
For an assignment $\alpha$,
we define the \REV{quasi-}{}precedence $\PGS^\alpha$ as follows:
$f \PGS^\alpha g$ iff $\alpha \models \PrecVarOf{f} \ge \PrecVarOf{g}$.

Next we consider representing \REV{a partial status
	by imitating the encoding of a \emph{filtered permutation}
	proposed in \cite{CGST12}.
	We introduce
}{
 statuses by
}%
 the boolean variables
$\PermedVarOf{f,i}$ and $\PermVarOf{f,i,j}$,
so that an assignment $\alpha$ induces a status $\sigma^\alpha$ as follows:
$\alpha \models \PermedVarOf{f,i}$ iff $i \in \sigma^\alpha(f)$, and
$\alpha \models \PermVarOf{f,i,j}$ iff $i$ is the $j$-th element in $\sigma^\alpha(f)$.
\REV{%
	In order for $\sigma^\alpha$ to be well-defined,
	every $i \in \sigma^\alpha(f)$ must occur exactly once in $\sigma^\alpha(f)$
	and any $i \notin \sigma^\alpha(f)$ must not occur in $\sigma^\alpha(f)$.
	These conditions are represented by
}{%
	In order to ensure $\sigma^\alpha$ to be well-defined,
	we introduce
}%
the following formula:
\[
	\ST \DefEq
	\BigAnd_{f\in\Sig_n}\BigAnd_{i=1}^n
	\Bigl(
		\Bigl( \PermedVarOf{f,i} \Then \sum_{j=1}^n \PermVarOf{f,i,j} = 1 \Bigr)\And
		\Bigl( \Not\PermedVarOf{f,i} \Then \sum_{j=1}^n \PermVarOf{f,i,j} = 0 \Bigr)
	\Bigr)
\]
It is easy to verify that $\sigma^\alpha$ is well-defined
if $\alpha \models \ST$.
\REV{%
	In contrast to the previous works \cite{ZHM09,CGST12}, we moreover
}{%
	Moreover we
}%
need the following formula to ensure
\REV{}{the }weak simplicity of $\A$ \wrt $\sigma$:
\[
	\SIMP \DefEq
	\BigAnd_{f \in \Sig_n} \BigAnd_{i = 1}^{n}
	\Bigl(\PermedVarOf{f,i} \Then \Encode{f(\Seq{x_n}) \AGS x_i}\Bigr)
\]
Note that in the formula $\SIMP$,
the condition $f(\Seq{x_n}) \AGS x_i$ can often be encoded in a more efficient way.
For linear $\Apol$, for example, this condition is equivalent to $\CoefOf{f,i} \ge 1$.

\begin{lemma}
	$\alpha \models \ST \And \SIMP$ iff
	$\A^\alpha$ is weakly simple \wrt a partial status $\sigma^\alpha$.\qed
\end{lemma}

Now we present the encodings for $\WPO$.

\begin{definition}
	The encodings of $\GT_\WPO$ and $\GS_\WPO$ are defined as follows:
	\[
		\Encode{s \GSopt_\WPO t} \DefEq
		\Encode{s \AGT t} \Or
		\bigl( \Encode{s \AGS t} \And s \GSopt_1 t \bigr)
	\]
	where the formula $s \GSopt_1 t$ is defined as follows:
	\begin{eqnarray*}
		x \GS_1 t &\DefEq&
		\begin{Cases}
			\True	&\text{if } x = t
		\\
			\False	&\text{otherwise}
		\end{Cases}
	\\
		x \GT_1 t &\DefEq& \False
	\\
		f(\Seq{s_n}) \GSopt_1 t &\DefEq&
		\bigvee_{i=1}^{n}
		\Bigl(\PermedVarOf{f,i} \And \Encode{s_i \GS_\WPO t} \Bigr) \Or f(\Seq{s_n}) \GSopt_2 t
	\end{eqnarray*}
	The formula $f(\Seq{s_n}) \GSopt_2 t$ is defined as follows:
	\begin{eqnarray*}
		f(\Seq{s_n}) \GSopt_2 y	&\DefEq& \False
	\\
		f(\Seq{s_n}) \GSopt_2 g(\Seq{t_m})	&\DefEq&
			\bigwedge_{j=1}^{m}
				\Bigl(\PermedVarOf{g,j} \Then \Encode{s \GT_\WPO t_j} \Bigr)
			\ \And\\
			&&
			\Bigl(
				\PrecVarOf{f} > \PrecVarOf{g}
				\Or
				\PrecVarOf{f} = \PrecVarOf{g} \And 
				\Encode{\ListOf{\Seq{s_{\sigma(f)}}} \GSopt_\WPO^\Lex \ListOf{\Seq{t_{\sigma(g)}}}}
			\Bigr)
	\end{eqnarray*}
\end{definition}
Here\REV{,
	$s \GSopt_1 t$ indicates that $s \GSopt_\WPO t$ is derived by
	case \prettyref{item:WPO-simp} or \prettyref{item:WPO-args}, and
	$s \GSopt_2 t$ indicates that $s \GSopt_\WPO t$ is derived by
	case \prettyref{item:WPO-prec} or \prettyref{item:WPO-mono}.
	We
}{ we }%
do not present the encoding for the lexicographic extension \wrt permutation,
which can be found in \cite{CGST12}.
\REV{\par}{}%
For an assignment $\alpha$, we write $\GSopt_\WPO^\alpha$ to denote 
the instance of WPO corresponding to $\alpha$;
\ie,
$\GSopt_\WPO^\alpha$ is induced by the algebra $\A^\alpha$,
the \REV{quasi-}{}precedence $\PGS^\alpha$ and the partial status $\sigma^\alpha$.

\begin{lemma}\label{lem:encode}
	For any assignment $\alpha$ \st
	$\alpha \models \ST \And \SIMP$,
	$\alpha \models \Encode{s \GSopt_\WPO t}$ iff
	$s \GSopt_\WPO^\alpha t$.\qed
\end{lemma}

\begin{theorem}\label{thm:encode DP WPO}
	If the following formula is satisfiable:
	\begin{equation}\label{eq:encode DP}
		\ST \And \SIMP \And
		\bigwedge_{l \to r \in \RR \cup \PP}\Encode{l \GS_\WPO r}\ \And
		\bigvee_{l \to r \in \PP'}\Encode{l \GT_\WPO r}
	\end{equation}
	then the DP processor that maps $\Tp{\PP,\RR}$ to $\{ \Tp{\PP\setminus\PP',\RR} \}$
	is sound.
\end{theorem}
\begin{proof}
	Let $\alpha$ be the assignment that satisfies \eqref{eq:encode DP}.
	By \prettyref{lem:encode}, we obtain
	$l \GS_\WPO^\alpha r$ for all $l \to r \in \RR \cup \PP$ and
	$l \GT_\WPO^\alpha r$ for all $l \to r \in \PP'$.
	Moreover by \prettyref{thm:WPO pair}, $\Tp{\GS_\WPO^\alpha, \GT_\WPO^\alpha}$
	forms a reduction pair.
	Hence \prettyref{thm:reduction pair} concludes the soundness of
	this DP processor.
	\qedhere
\end{proof}

In the following sections,
we give encodings depending on the choice of $\A$ for each instance of $\WPO$.

\subsection{Encoding WPO$(\Apol)$ and WPO$(\Asum)$}

First we present \REV{an encoding of a}{encodings for} linear polynomial interpretation $\Apol$.
The encodings for $\Asum$ is obtained by fixing \REV{all }{}coefficients to $1$.
The weight of a term $s$ and the variable coefficient of $x$ in $s$
are encoded as follows:

\begin{align*}
	\WsumOf{s} &\DefEq
	\begin{Cases}
	\WeiVarZero &\text{ if } s \in \Vars \\
	\WeiVarOf{f} + 
	\displaystyle\sum_{i = 1}^{n}\CoefVarOf{f,i} \cdot \WsumOf{s_i} &\text{ if } s = f(\Seq{s_n})
	\end{Cases}
\\
	\VCoefOf{x,s} &\DefEq
	\begin{Cases}
		1 &\text{if } x = s
	\\
		0 &\text{if } x \neq s \in \Vars
	\\
		\displaystyle\sum_{i=1}^{n}\CoefVarOf{f,i} \cdot \VCoefOf{x,s_i}
		&\text{if } s = f(\Seq{s_n})
	\end{Cases}
\end{align*}

\def\COEF{\mathsf{COEF}}
\def\WMIN{\mathsf{WMIN}}
\REV{%
	We have to ensure $\WeiVarZero$ to be the lower bound of weights of terms.
	To ensure $\WsumOf{f(\Seq{s_n})} \ge \WeiVarZero$
	for every term $f(\Seq{s_n})$,
	we need either $\WeiVarOf{f} \ge \WeiVarZero$ or
	one of the arguments to have a positive coefficient
	(note that the weight of this argument is at least $\WeiVarZero$).
	This is represented by
}{%
	In order to ensure $\Wzero$ to be the lower bound,
	we introduce
}%
the following constraint:
\[
	\WMIN \DefEq \BigAnd_{f \in \Sig_n}
	\Bigl( \WeiVarOf{f} \ge \WeiVarZero \Or \BigOr_{i=1}^n \CoefVarOf{f,i} \ge 1
	\Bigr)
\]

Now the relations $>_\Apol$ and $\ge_\Apol$ are encoded as follows:
\[
	\Encode{s \geopt_\Apol t} \DefEq
	\WsumOf{s} \geopt \WsumOf{t} \And
	\bigwedge_{x \in \Var(t)} \VCoefOf{x,s} \geq \VCoefOf{x,t}
\]

\begin{corollary}
	If the following formula is satisfiable:
	\[
		\ST \And \SIMP \And \WMIN \And
		\bigwedge_{l \to r \in \RR \cup \PP}\Encode{l \GS_\WPOpol r}\ \And
		\bigvee_{l \to r \in \PP'}\Encode{l \GT_\WPOpol r}
	\]
	then the DP processor that maps $\Tp{\PP,\RR}$ to $\{ \Tp{\PP\setminus\PP',\RR} \}$
	is sound.
	\qed
\end{corollary}

\subsection{Encoding WPO$(\Amax)$}

In this section, we consider encoding $\WPOmax$.
Unfortunately, we are aware of no SMT solver which supports a built-in $\max$ operator.
Hence we consider encoding the constraint $s >_\Amax t$ into
both quantified and quantifier-free formulas.

First, we present an encoding to a quantified formula.
A straightforward encoding would involve
\[
	\WmaxOf{s}\DefEq
	\begin{Cases}
		s &\text{if } s \in \Vars\\
		v &\text{if } s = f(\Seq{s_n})
	\end{Cases}
\]
where $v$ is a fresh integer variable
\REV{%
	representing $\max\{ \WeiVarOf{f}, \WmaxOf{s_1},\dots,\WmaxOf{s_n} \}$
}{%
}%
with the following constraint $\phi$ added into the context:
\begin{align*}
	\phi \DefEq{}&
		v \ge \WeiVarOf{f} \And \bigwedge_{i=1}^{n} v \ge \WmaxOf{s_i}
	\And
	\Big(
		v = \WeiVarOf{f} \Or \bigvee_{i=1}^{n} v = \WmaxOf{s_i}
	\Big)
\end{align*}
Then the constraint $s \geopt_\Amax t$ can be encoded as follows:
\[
	\Encode{s \geopt_\Amax t}\DefEq
	\ForAll{\Seq{x_k}, \Seq{v_m}}
	\phi_1 \And \dots \And \phi_m \Then
	\WmaxOf{s} \geopt \WmaxOf{t}
\]
where $\{ x_1, \dots, x_k \} = \Var(s) \cup \Var(t)$ and
each $\Tp{\phi_j,v_j}$ is
the pair of the constraint and the fresh variable
introduced during the encoding.

Although quantified linear integer arithmetic is known to be decidable,
the SMT solvers we have tested could not solve the problems 
generated by the above straightforward encoding efficiently, if at all.
Fuhs \etal \cite{FGMSTZ08} propose\REV{}{s} a sound elimination of quantifiers
by introducing new template polynomials.
Here we propose another encoding to quantifier-free formulas
that \REV{does not introduce extra polynomials and
}{}is sound and complete for linear polynomials with max.

\begin{definition}
	A \emph{generalized weight} \cite{KV03b} is a pair $\Tp{n,N}$
	where $n \in \Nat$ and $N$ is a finite multiset%
\footnote{In the encoding for $\Amax$, $N$ need not contain more than one variable.
This generality is reserved for \REV{the }{}encoding of $\Amp$.}
	over $\Vars$.
	We define the following operations:
	\begin{align*}
		\Tp{n,N} + \Tp{m,M} &\DefEq \Tp{n + m, N \uplus M}
	\\
		n \cdot \Tp{m,M} &\DefEq \Tp{n \cdot m, n \cdot M}
	\end{align*}
	where $n \cdot M$ denotes the multiset that maps 
	$x$ to $n \cdot M(x)$ for every $x \in \Vars$.
	We encode a generalized weight as a pair of an expression and
	a mapping $N$ from $\Vars$ to expressions \st
	the \emph{domain} $\Dom(N) \DefEq \{ x \mid N(x) \neq 0 \}$ of $N$ is finite.
	Notations for generalized weights are naturally extended for encoded ones.
	The relation $\supseteq$ on multisets is encoded as follows:
	\[
		N \supseteq M \DefEq \bigwedge_{x \in \Dom(M)} N(x) \ge M(x)
	\]
\end{definition}
A generalized weight $\Tp{n,N}$ represents the expression
$n + \sum_{x\in N}x$.
Now we consider removing $\max$.

\begin{definition}\label{def:Wmax}
	The \emph{expanded weight} $\XWof{s}$ of a term $s$ induced by
	a weight function $\Tp{\Weight,\Wzero}$ and a subterm penalty function
	$\Pen$ is a set of generalized weights, which is defined as follows:
	\[
		\XWof{s} \DefEq
		\begin{Cases}
			\{ \Tp{\WeiVarZero, \SetOf{s}} \} &\text{if } s \in \Vars
		\\
			\{ \Tp{\WeiVarOf{f}, \Set\Empty} \} \cup
			\{ \PenVarOf{f,i} + p \mid p \in \XWof{s_i} \text{, } 1 \le i \le n \}
			&\text{if } s = f(\Seq{s_n})
		\end{Cases}
	\]
\end{definition}
The expanded weight $\XWof{s} = \{ \Seq{p_n} \}$ represents
the expression $\max \{ \Seq{e_n} \}$,
where each generalized weight $p_i$ represents the expression $e_i$.
Using expanded weights, we can encode $>_\Amax$ and $\ge_\Amax$
in a way similar to the \emph{max set ordering} presented in \cite{BC08}:
\[
	\Encode{s \geopt_\Amax t} \DefEq
	\bigwedge_{\Tp{m,M} \in \XWof{t}}
	\bigvee_{\Tp{n,N} \in \XWof{s}}
	(n \geopt m \And N \supseteq M)
\]

Using the quantified or quantifier-free encodings,
we obtain the following corollary of \prettyref{thm:WPO pair}:

\begin{corollary}
	If the following formula is satisfiable:
	\[
		\ST \And
		\bigwedge_{l \to r \in \RR \cup \PP}\Encode{l \GS_\WPOmax r}\ \And
		\bigvee_{l \to r \in \PP'}\Encode{l \GT_\WPOmax r}
	\]
	then the DP processor that maps $\Tp{\PP,\RR}$ to $\{ \Tp{\PP\setminus\PP',\RR} \}$
	is sound.
	\qed
\end{corollary}

\subsection{Encoding WPO$(\Amp)$ and WPO$(\Ams)$}

In this section, we consider encoding linear polynomials with max into SMT formulas.
First we extend \prettyref{def:Wmax} for weight statuses.

\begin{definition}
	For a weight status $\Wstatus$,
	the \emph{expanded weight} $\XWsmOf{s}$ of a term $s$ is the set of
	generalized weight, which is recursively defined as follows:
	\begin{align*}
		\XWsmOf{s} &\DefEq
		\begin{Cases}
			\{ ( \WeiVarZero, \{ s \} ) \}&\text{if } s \in \Vars
		\\
			S&\text{if } s = f(\Seq{s_n}),\ \WstatusOf{f} = \Pol
		\\
			T&\text{if } s = f(\Seq{s_n}),\ \WstatusOf{f} = \Max
		\end{Cases}
	\end{align*}
	where
	\begin{align*}
		S &= \Bigl\{ \WeiVarOf{f} + \sum_{i=1}^{n} \CoefVarOf{f,i} \cdot p_i\ \big|\ 
		p_1 \in \XWsmOf{s_1}, \dots, p_n \in \XWsmOf{s_n} \Bigr\}
	\\
		T &= \{ \WeiVarOf{f} \} \cup 
		\{
			\PenVarOf{f,i} + \CoefVarOf{f,i} \cdot p \mid p \in \XWsmOf{s_i} \text{, }
			i \in \SetOf{1,\dots,n}
		\}
	\end{align*}
\end{definition}

Now the encoding of $>_\Amp$ and $\ge_\Amp$ are given as follows:
\[
	\Encode{s \geopt_\Amp t} \DefEq
	\bigwedge_{\Tp{m,M} \in \XWsmOf{t}}
	\bigvee_{\Tp{n,N} \in \XWsmOf{s}}
	\bigl( n \geopt m \And N \supseteq M \bigr)
\]

\begin{corollary}
	If the following formula is satisfiable:
	\[
		\ST \And \SIMP \And \WMIN \And
		\bigwedge_{l \to r \in \RR \cup \PP}\Encode{l \GS_\WPOmp r}\ \And
		\bigvee_{l \to r \in \PP'}\Encode{l \GT_\WPOmp r}
	\]
	then the DP processor that maps $\Tp{\PP,\RR}$ to $\{ \Tp{\PP\setminus\PP',\RR} \}$
	is sound.
	\qed
\end{corollary}

\subsection{Encoding WPO$(\Amat)$}

We omit presenting an encoding of the matrix interpretation method,
which can be found in \cite{EWZ08}.
In order to use a matrix interpretation in WPO,
however, \REV{}{a }small care is needed; one \REV{has}{have} to ensure
weak simplicity of $\Amat$ \wrt $\sigma$.
This can be done as follows:

\begin{lemma}
	If
	$\CmatOf{f,i}^{j,j} \ge 1$ for all
	$f \in \Sig_n$, $i \in \sigma(f)$ and
	$j \in \SetOf{1, \dots, d}$,
	then $\Amat$ is weakly simple \wrt $\sigma$.\qed
\end{lemma}

\subsection{Encoding for Reduction Orders}\label{sec:encode order}

\def\TOTAL{\mathsf{TOTAL}}

In case one wants an encoding for the reduction order $>_\WPO$ defined in
\prettyref{def:WPO}, then the status $\sigma$ must be total.
This can be ensured by
\REV{%
	enforcing $i \in \sigma(f)$ for all $i \in \{ 1, \dots, n \}$ and $f \in \Sig$,
	which is represented by
}{%
}%
the following formula:
\[
	\TOTAL \DefEq \BigAnd_{f\in\Sig_n}\BigAnd_{i=1}^{n}\PermedVarOf{f,i}
\]
or equivalently by replacing all $\PermedVarOf{f,i}$ by $\True$.
Note that $\TOTAL \And \SIMP$
enforces all the subterm coefficients to be greater than
\REV{%
	or equal to
}{}%
$1$.

\begin{theorem}\label{thm:encode order}
	If the following formula is satisfiable:
	\[	\TOTAL \And \ST \And \SIMP \And \WMIN \And \BigAnd_{l \to r \in \RR} \Encode{l \GT_\WPOmp r}
	\]
	then $\RR$ is orientable by $\WPOmp$.\qed
\end{theorem}

\section{Optimizations}\label{sec:optimizations}

In our implementation, some
optimizations are performed during the encoding.
For example,
formulas like $\False \And \phi$ are reduced in advance 
to avoid generating meaningless formulas, and
temporary variables are inserted to avoid multiple occurrences of
an expression or a formula.
Moreover, we apply several optimizations that we discuss below.

\subsection{Fixing $\Wzero$}\label{sec:fixing w0}

We can simplify the encoded formulas by fixing $\Wzero$.
For KBO, Winkler \etal \cite{WZM12} show\REV{}{s} that
$\Wzero$ can be fixed to \REV{an }{}arbitrary $k > 0$ \REV{(\eg, $1$)}{\eg $1$}
without \REV{losing any power}{loosing the power of the order}.
\REV{%
	By adapting the proof of \cite[Lemma 3]{WZM12},
	it can be shown that $\Wzero$ can be fixed to $0$ for $\WPOsum$.
}{%
	Applying their technique,
	it can be shown that for $\WPOsum$, $\Wzero$ can be fixed to $0$.
}%
On the contrary to KBO, however,
\REV{%
	fixing $\Wzero > 0$ will affect the power,
	since the transformation of \cite{WZM12} may assign
	negative weights to some symbols
	when applied to the case $\Wzero = 0$.
}{%
	$\Wzero$ cannot be fixed to $k > 0$
	since transforming a weight function $\Tp{\Weight,0}$ into $\Tp{\Weight^k,k}$
	may assign negative weights to some symbols.
}%

\subsection{Fixing Weight Status}

For POLO and WPO using algebras $\Ams$ and $\Amp$,
it may not be practical to consider all possible weight statuses.
Hence, we introduce a heuristic for fixing $\Wstatus$.
In case of $\WPOms$, 
$\Wstatus$ should at least satisfy the following condition
for all $l \to r \in \RR \cup \PP$:
\[
	\BigAnd_{\Tp{m,M} \in \XWsmOf{r}}
	\BigOr_{\Tp{n,N} \in \XWsmOf{l}}
	N \supseteq M
\]
since otherwise the formula
$\BigAnd_{l \to r \in \RR} \Encode{l \GS_\WPOms r}$
is trivially unsatisfiable.
Hence in our implementation,
we \REV{require that}{consider} $\Ams$ and $\Amp$ are induced by
the weight status which minimizes 
the number of $f$ with $\WstatusOf{f} = \Max$,
while satisfying the above condition.

\subsection{Reducing Recursive Checks}\label{sec:reduce recursion}

Encoding KBO as a reduction order \cite{ZHM09}
is notably efficient,
because KBO does not have \REV{}{a }recursive checks like LPO or WPO.
For WPO, we can reduce formulas for recursive checks
by restricting $w_0 > 0$, since under \REV{this}{the} restriction\REV{,}{}
$f(\Seq{s_n}) >_\Amp s_i$ \REV{holds }{}whenever $n \ge 2$ and $\WstatusOf{f} = \Pol$.
Hence if $n \ge 2$ and $\WstatusOf{f} = \Pol$,
we reduce the formula $f(\Seq{s_n}) \GSopt_1 t$ to $f(\Seq{s_n}) \GSopt_2 t$.
Analogously if $m \ge 2$ and $\WstatusOf{g} = \Pol$,
we reduce the formula $f(\Seq{s_n}) \GSopt_2 g(\Seq{t_m})$ to
the following:
\[
	\PrecVarOf{f} > \PrecVarOf{g} \Or
	\PrecVarOf{f} = \PrecVarOf{g} \And
	\Encode{\AppPerm{\sigma(f)}{s}{n} \GSopt_\WPOAS^\Lex
			\AppPerm{\sigma(g)}{t}{m}
	}
\]
without generating formulas for recursive checks corresponding to
cases \prettyref{item:WPO-simp} and \prettyref{item:WPO-args}.
Note \REV{however that}{that however,} this simplification does not apply
when encoding reduction pairs using argument filtering,
as we will see in \prettyref{sec:experiments pair}.

\section{Experiments}\label{sec:experiments}

In this section we examine the performance of WPO
both as a reduction order and in the DP framework.
\REV{%
	We implemented a simple form of the DP framework as the 
	\emph{Nagoya Termination Tool} (\NaTT)
	and incorporated WPO as a DP processor \cite{YKS14b}.
	The
}{We implemented the}%
 encodings presented in \prettyref{sec:encodings}
and optimizations presented in \prettyref{sec:optimizations}%
\REV{ are implemented}{}.
In the encodings of $\Apol$ and $\Amp$,
we choose $3$ \REV{as upper bound}{for upper bounds} of weights and coefficients\REV{,
	in order to achieve a practical runtime}{}.
For comparison, KBO, TKBO, LPO,
polynomial interpretations with or without max
and matrix interpretations are
implemented in the same manner.
For the DP framework,
we implemented
the estimation of \emph{dependency graphs} in \cite{GTS05},
and \emph{strongly connected components} are sequentially processed
in order of size where smaller ones \REV{come first}{are precedent}.
Moreover,
\emph{usable rules} \wrt argument filters are also implemented by
following the encoding proposed in \cite{CSLTG06}.

The test set of termination problems are 
the 1463 TRSs from the TRS Standard category of TPDB 8.0.6 \cite{TPDB13}.
The experiments are run on a server equipped with 
two quad-core Intel Xeon W5590 processors running at a clock rate of 3.33GHz
and 48GB of main memory, though only one thread of SMT solver runs at once.
As the SMT solver, we choose \Zthreevar{4.3.1}.%
\footnote{\url{http://z3.codeplex.com/}}
Timeout is set to 60s, as in the \emph{Termination Competition} \cite{TC13}.
Details of the experimental results are available at
\url{http://www.trs.cm.is.nagoya-u.ac.jp/papers/SCP2014/}.

\subsection{Results for Reduction Orders}

First we evaluated WPO as a reduction order
by directly testing orientability for input TRSs.
The results are listed in \prettyref{tab:order}.
Since KBO, POLO($\Asum$), WPO($\Asum$) are only applicable for
non-duplicating TRSs,
the test set is split into
non-duplicating ones (consisting of 439 TRSs) and
duplicating ones (consisting of 1024 TRSs).
In the table,
\REV{the }{}`yes' column indicates the number of successful termination proofs,
`T.O.' indicates the number of timeouts, and `time' indicates the total time.
\REV{
	To emphasize the benefit of WPO,
	we also compare it with arbitrary lexicographic compositions of
	POLO, KBO, and LPO (`POLO+KBO+LPO' row).
}{}%
\begin{table}[tb]
\caption{Results for Reduction Orders\label{tab:order}}
\centering
\begin{tabular}{cc@{\quad}rc@{\ }rl@{\quad}rc@{\ }r}
\hline
	&&\multicolumn{3}{c}{non-dup. TRSs}
	&&\multicolumn{3}{c}{dup. TRSs}\\
\cline{3-5}\cline{7-9}
	order&algebra	&yes&{T.O.}&time&&yes&{T.O.}&time\\
\hline
	POLO&$\Asum$	&41	&0	&4.45	&&--	&--	&--	\\
	POLO&$\Ams$		&60	&0	&4.46	&&19	&0	&	\\
	LPO	&			&90	&0	&31.64	&&90	&0	&35.39\\
	KBO	&			&115&0	&6.20	&&--	&--	&--\\
\hdashline
	WPO	&$\Asum^+$	&126&0	&6.70	&&--	&--	&--		\\
	WPO	&$\Asum$	&\bf 135&0	&43.16	&&--	&--	&--	\\
	WPO	&$\Amax$	&109&0	&53.77	&&125	&0	&49.49	\\
	WPO	&$\Ams$		&\bf 135&0	&42.72	&&\bf 138	&0	&66.15\\
\hline
	POLO&$\Apol$	&104&3	&203.37	&&21	&10	&1065.34\\
	POLO&$\Amp$		&104&3	&203.07	&&39	&8	&608.76	\\
	TKBO&			&132&3	&226.33	&&27	&12	&1414.62\\
\hdashline
	WPO	&$\Apol$	&\bf 149&3	&280.92	&&29	&12	&1495.55\\
	WPO	&$\Amp$		&\bf 149&3	&280.86	&&\bf 138	&9	&1008.67\\
\hline
	\multicolumn{2}{c}{\REV{POLO+KBO+LPO}{}}
					&\REV{130}{}&\REV{0}{}	&\REV{35.35}&&\REV{92}{}&\REV{0}{}&\REV{51.35}{}\\
\hline
\end{tabular}
\end{table}

We point \REV{out }{}that $\WPO(\Ams)$ is a balanced choice;
it is significantly stronger than existing orders,
while the runtime is much better than
involving non-linear SMT solving (last 5 columns).
Note that we directly solve non-linear problems using \Zthree;
\REV{it}{It} may be possible to improve efficiency by
\REV{\eg, a}{\eg} SAT encoding like \cite{FGMSTZ07}.
In that case, we expect $\WPO(\Amp)$ to become a practical choice.
If \REV{}{the }efficiency is the main concern, then
$\WPO(\Asum^+)$, a variant of $\WPOsum$ with $\Wzero > 0$,
is a reasonable substitute for KBO.
This efficiency is due to the reduction of recursive checks
proposed in \prettyref{sec:reduce recursion}.

\subsection{Results for Reduction Pairs}\label{sec:experiments pair}

Second,
we evaluated WPO as a reduction pair.
\prettyref{tab:pair} compares the power of the reduction pair processors.
In `total status' column,
we apply standard total statuses and argument filtering to
obtain a reduction pair from a reduction order, as in
\cite{YKS13}.
Because of argument filtering,
\REV{the }{}existence of duplicating rules \REV{is}{are} not an issue in this setting.
For POLO, statuses and argument filtering are ignored
(the latter is considered as $0$-coefficient).
\begin{table}[tb]
\caption{Results for Reduction Pairs\label{tab:pair}}
\centering
\begin{tabular}{cccc@{\quad}c@{\ }rlc@{\quad}c@{\ }r}
\hline
	&&\qquad&\multicolumn{3}{c}{total status}
	&\quad&\multicolumn{3}{c}{partial status}
\\\cline{4-6}\cline{8-10}
	order&algebra&&yes&T.O.&\multicolumn{1}{c}{time}
		&&yes&T.O.&\multicolumn{1}{c}{time}
\\\hline
	POLO&$\Asum$	&&512&0	&150.99	&&--	&--	&--\\
	POLO&$\Ams$		&&522&0	&300.65	&&--	&--	&--\\
	LPO	&			&&502&0	&435.62	&&--	&--	&--\\
	KBO	&			&&497&3	&1001.55&&520	&4	&1238.67\\
\hdashline
	WPO	&$\Asum$	&&514&3	&907.27	&&560	&5	&1244.04\\
	WPO	&$\Amax$	&&548&7	&1269.48&&637	&13	&1846.06\\
	WPO	&$\Ams$		&&\bf 578&5	&1261.63&&\bf 675	&12	&1827.01\\
\hline
	POLO&$\Apol$	&&544&19	&1958.44	&&--	&--	&--\\
	POLO&$\Amp$		&&540&18	&1889.86	&&--	&--	&--\\
	POLO&$\Amat$	&&\bf 645&480	&32367.26	&&--	&--	&--\\
	TKBO&			&&516&187	&15665.26	&&539	&178&15799.28\\
\hdashline
	WPO	&$\Apol$	&&527&172	&14535.24	&&579	&153&13579.95\\
	WPO	&$\Amp$		&&560&88	&7678.43	&&\bf 672	&94	&9269.36\\
	WPO	&$\Amat$	&&-- &--	&--			&&538	&640&42067.45\\
\hline
	\REV{\THOR}{}&	&&\REV{418}{}&\REV{261}{}&\REV{18550.62}{}&&--	&--	&--\\
\hline
\end{tabular}
\end{table}
The power of WPO is still measurable here. On
the contrary to the reduction order case,
$\WPOsum$ outperforms KBO both in power and efficiency.
This is because 
KBO needs formulas for recursive comparison that resembles WPO,
when argument filters are considered.
Moreover,
encodings of weights are more complex in KBO, since $w_0$ cannot
be fixed to $0$ as discussed in \prettyref{sec:fixing w0}.
Finally, KBO needs extra constraints that correspond to \REV{}{the }admissibility.

In `partial status' column, we moreover admit partial statuses
of \prettyref{def:WPOpS}.
We also apply partial status for KBO as in \cite{YKS13b} but
not for LPO,
since LPO does not benefit from partial statuses because
weights are not considered.
The power of WPO is much more significant in this setting, and
$\WPOms$ is about 30\% stronger than any other existing techniques.
Though the efficiency is sacrificed for partial statuses,
this is not a severe problem in the DP framework,
as we will see in the next section.
On the other hand, our implementation of
the instances of WPO that require non-linear SMT solving are extremely time-consuming.
Especially, $\WPOmat$ \REV{loses}{looses} 107 problems by timeout compared to the
standard matrix interpretation method.
We conjecture that the situation can be improved by SAT encoding or 
by using other non-linear SMT solvers such as \cite{ZM10,BLNRR12}.
\REV{%

	In order to estimate the power of RPOLO,
	we also ran an experiment with \THOR,%
	\footnote{\url{http://www.lsi.upc.edu/~albert/term.html}}
	the only termination prover having RPOLO implemented, as far as we know.
	Note however that it might be unfair to compare the results directly;
	\THOR is specialized to higher-order case and is based on MSPO,
	while our implementation is based on the DP framework.
}{}%

\subsection{Combining DP Processors}\label{sec:experiments combination}

\REV{}{Finally, we evaluate WPO in a more practical use for termination provers.}%
In \REV{}{the }modern termination provers,
DP processors are combined in the DP framework and
weak but efficient ones are applied first.
\REV{%
	The default strategy of \NaTT sequentially applies
	the \emph{rule removal processor} \cite{GTS04},
	the \emph{(generalized) uncurrying} \cite{HMZ13,ST11},
	reduction pair processors including 
	standard POLO, LPO, POLO with max \cite{FGMSTZ08} and
	WPO($\Ams$) with partial status,
	and then a simple variant of the matrix interpretation method \cite{EWZ08}.
	When all reduction pair processors fail,
	a naive loop detection is performed to conclude nontermination.
	In \prettyref{tab:combination},
	we compare the following settings:
	`\NaTT' (the default strategy described above),
	`\NaTT w/o WPO' (WPO is replaced by KBO),%
	\footnote{However, KBO does not contribute in this strategy.}
	\AProVEver{2014}, and \TTTTver{1.15}.
	The `no' column indicates the number of successful nontermination proofs.
}{%
	In \prettyref{tab:combination},
	we compare two strategies for combining DP processors.
	For `existing' strategy, we sequentially apply
	the \emph{rule removal processor} \cite{GTS04},
	the \emph{(generalized) uncurrying} \cite{HMZ13,ST11},
	and then reduction pair processors including 
	standard POLO, LPO,
	POLO with max \cite{FGMSTZ08},
	KBO,\footnote{However, KBO does not contribute in this strategy.}
	and then a simpler variant of the matrix interpretation method \cite{EWZ08}.
	In `new' strategy, we replace KBO by $\WPOms$ with partial status.
}%
\begin{table}[tb]
\caption{Results for Combination\label{tab:combination}}
\centering
\begin{tabular}{cccccr}
\hline
	\REV{tools}{strategy}
				&yes		&\REV{no}{}	&\REV{maybe}{}&T.O.		&\multicolumn{1}{c}{time}\\
\hline
	\REV{\NaTT}{new}
				&848		&\REV{173}{}&\REV{429}{}&\REV{13}{10}	&18\REV{65.50}{36.30}\\
	\REV{\NaTT w/o WPO}{existing}
				&810		&\REV{173}{}&\REV{467}{}&\REV{13}{10}	&\REV{2023.18}{1714.07}\\
	\REV{\AProVE}
				&\REV{1020}{}&\REV{271}{}&\REV{0}{}&\REV{173}{}		&\REV{15123.48}{}\\
	\REV{\TTTT}	&\REV{788}{}&\REV{193}{}&\REV{417}{}&\REV{65}{}		&\REV{13784.43}{}\\
\hline
\end{tabular}
\end{table}%
In this setting, \REV{\NaTT}{our tool} discovered
termination proofs for \REV{36}{40} of \REV{159}{161} problems
whose termination
\REV{%
	could not be proved by any other tools participated
}{%
	are unknown
}%
in the full-run of the Termination Competition
\REV{%
	2013 \cite{TC13}. For 29 of these problems,
}{%
	2011 \cite{TC13}, and for 29 of them,
}%
WPO is essential.%
\footnote{Due to the efficiency of our implementation,
our tool proves \REV{7}{11} open problems in TPDB without using WPO.}
\REV{%
	In the competition,
	\NaTT finished in the remarkable second place in the TRS standard category.
}{}%

\section{Conclusion}\label{sec:conclusion}

We introduced the weighted path order both as a reduction order 
and as a reduction pair.
We presented several instances of WPO as reduction orders:
$\WPOsum$ that subsumes KBO,
$\WPOpol$ that subsumes POLO and TKBO,
$\WPOmax$ that subsumes LPO,
$\WPOms$ that unifies KBO and LPO, and
$\WPOmp$ that unifies all of them.
Moreover, we applied partial status for WPO to obtain a reduction pair,
and presented further refinements.
We show that as a reduction pair, WPO subsumes
KBO, LPO and TKBO with argument filters, POLO, and matrix interpretations.
We also presented SMT encodings for these techniques.
The orientability problems of $\WPOsum$, $\WPOmax$ and $\WPOms$ are decidable,
since they are reduced to satisfiability problems of linear integer arithmetic
which is known to be decidable.
Finally, we verified through experiments the significance of our work
both as a reduction order and as a reduction pair.
In order to keep the presentation simple,
we did not present $\WPO$ with \emph{multiset status}.
Nonetheless, it is easy to define $\WPO$ with multiset status and verify that
$\WPOmax$ with multiset status \REV{subsumes}{encompasses} RPO.

We only considered a straightforward method for combining $\WPOpol$ and $\WPOmax$
using `weight statuses', and moreover heuristically fixed the weight status.
We leave it for future work to
search for other possible weight statuses, or to
find more sophisticated combination\REV{s}{} of max-polynomials
such as $f_\A(x,y,z) = x + \max(y,z)$, or
even trying other algebras including
\emph{ordinal interpretations} \cite{WZM12,WZM13}.
For efficiency, real arithmetic is also attractive to consider, since
SMT for \REV{}{the }real arithmetic is often more efficient than for
\REV{}{the }integer arithmetic.
To this end, we will have to reconstruct the proof \REV{of}{for}
well-foundedness of WPO, since our current proof relies on
well-foundedness of the underlying order,
which does not hold anymore for real numbers.
Another obvious future work is to extend WPO for higher\REV{-}{ }order case.
Since RPOLO has strength in its higher\REV{-}{ }order version \cite{BBRR13},
we expect their technique can be extended for WPO.

\paragraph*{Acknowledgments}

\REV{%
	We are grateful to the anonymous reviewers for their careful inspections
	and comments that significantly improved the presentation of this paper.
	We thank Sarah Winkler and Aart Middeldorp for discussions at
	the early stages of this work.
	We thank Florian Frohn and J\"urgen Giesl for
	their helps in experiments with \AProVE, and
	Albert Rubio, Miquel Bofill and Cristina Borralleras for
	their helps in experiments with \THOR.
}{%
	We thank Sarah Winkler and Aart Middeldorp for discussions
	and the anonymous reviewers of previous versions of this paper
	for helpful comments.
}%
This work was supported by JSPS KAKENHI \#24500012.%
\REV{%
\appendix
\section{Omitted Proofs}

	In this appendix, we prove several properties that
	are needed for soundness of WPO.
}{%
	In the rest of this section,
	we prove several properties that \REV{are}{is} needed for
	\prettyref{def:WPOpS} to define a correct reduction pair.
}%
\REV{The first}{First} one is obvious from the definition.

\begin{lemma}\label{lem:WPO strict}
	${\GT_\WPO}\subseteq{\GS_\WPO}$.\qed
\end{lemma}

Using the above lemma, we show compatibility of $\GS_\WPO$ and $\GT_\WPO$.

\begin{lemma}[Compatibility]\label{lem:WPO compatible}
	$\GT_\WPO$ is compatible with $\GS_\WPO$.
\end{lemma}
\begin{proof}
	Supposing $s \GS_\WPO t \GT_\WPO u$, we
	show $s \GT_\WPO u$ by induction on $\Tp{|s|,|t|,|u|}$.
	The other case, $s \GT_\WPO t \GS_\WPO u$, is analogous.

	From the definition,
	it is obvious that
\REV{%
	$t$ is of the form $g(\Seq{t_m})$
}{%
	$s$ and $t$ are 
	\REV{of the }{in }form $f(\Seq{s_n})$ and $g(\Seq{t_m})$, \resp,
}%
	and moreover $s \AGS t \AGS u$.
	If $s \AGT t$ or $t \AGT u$, then
	we obtain $s \GT_\WPO u$ by case \prettyref{item:WPO-gt}.
\REV{%
	If $s \in \Vars$, then only case \prettyref{item:WPO-min} is applicable for 
	$s \GT_\WPO t$.
	In this case, $t \GT_\WPO u$ is not possible,
	since $\sigma(g) = \List\Empty$ and $g \PNGT h$ for any $h$.
	Hence, $s$ is of the form $f(\Seq{s_n})$.
}{}%
	If $s_i \GS_\WPO t$ for some $i \in \sigma(f)$, then
	by the induction hypothesis and \prettyref{lem:WPO strict}, we get $s_i \GS_\WPO u$
	and hence case \prettyref{item:WPO-simp} applies for $s \GT_\WPO t$.
	Now suppose $s \GT_\WPO t_j$ for every $j \in \sigma(g)$ and
	either $f \PGT g$ or $f \PSIM g$ and
	$\AppPerm{\sigma(f)}{s}{n} \GS_\WPO^\Lex \AppPerm{\sigma(g)}{t}{m}$ holds.
	There remain the following cases to consider for $t \GT_\WPO u$:
	\begin{itemize}
	\item
		$t_j \GS_\WPO u$ for some $j \in \sigma(g)$.
		In this case, we have $s \GT_\WPO t_j \GS_\WPO u$.
		Hence\REV{, by}{} the induction hypothesis
		\REV{we conclude}{concludes} $s \GT_\WPO u$.
	\item
		$u = h(\Seq{u_l})$ and
		$t \GT_\WPO u_k$ for every $k \in \sigma(h)$.
		By the induction hypothesis, we obtain $s \GT_\WPO u_k$.
		Moreover, if $f \PGT g$ or $g \PGT h$, then
		we get $f \PGT h$ and \prettyref{item:WPO-prec} applies.
		Otherwise, we have $f \PSIM g \PSIM h$ and
		\(
			\AppPerm{\sigma(f)}{s}{n} \GS_\WPO^\Lex
			\AppPerm{\sigma(g)}{t}{m} \GT_\WPO^\Lex
			\AppPerm{\sigma(h)}{u}{k}
		\).
		From the induction hypothesis, we obtain
		$\AppPerm{\sigma(f)}{s}{n} \GT_\WPO^\Lex \AppPerm{\sigma(h)}{u}{k}$.
		Hence \prettyref{item:WPO-mono} applies.
		\qedhere
	\end{itemize}
\end{proof}

In order to prove well-foundedness of $\GT_\WPO$,
we define the set $\SN$ of \emph{strongly normalizing} terms
as follows:
$s \in \SN$
iff there exist\REV{s}{} no infinite reduction
sequence $s \GT_\WPO s_1 \GT_\WPO s_2 \GT_\WPO \dots$ beginning from $s$.
In the next two lemmas, we prove that all terms are in $\SN$.

\begin{lemma}\label{lem:WPO compose}
	Suppose that $\A$ is weakly simple \wrt $\sigma$.
	If $s = f(\Seq{s_n})$ with
	$s_i \in \SN$ for every $i \in \sigma(f)$, then
	$s \GT_\WPO t$ implies
	$t \in \SN$.
\end{lemma}
\begin{proof}
	We perform
	induction on $\Tp{s, f, \AppPerm{\sigma(f)}{s}{n}, |t|}$
	which is ordered by the lexicographic composition of
	$\AGT$, $\PGT$, $\GT_\WPO^\Lex$ and $>$.
	Since \REV{the claim}{it} is obvious if $t \in \Var$,
	we consider $t = g(t_1,\dots,t_m)$.
	\begin{enumerate}
	\item Suppose $s \AGT t$.
		First we show that $t_j \in \SN$ for every $j \in \sigma(g)$.
		By the weak simplicity assumption, we have $s \AGT t \AGS t_j$ and hence
		$s \GT_\WPO t_j$.
		Thus by the induction hypothesis on the fourth component, we obtain $t_j \in \SN$.
		Now for arbitrary $u$ \st $t \GT_\WPO u$,
		the induction hypothesis on the first component yields $u \in \SN$.
	\item Suppose $s \AGS t$.
		There are two subcases to consider.
		\begin{enumerate}
		\item Suppose $s_i \GS_\WPO t$ for some $i \in \sigma(f)$.
			Then by the assumption, we have $s_i \in \SN$.
			Hence by \prettyref{lem:WPO compatible}, we obtain
			$t \in \SN$.
		\item Suppose $s \GT_\WPO t_j$ for every $j \in \sigma(g)$.
			Then by the induction hypothesis on the fourth component,
			we get $t_j \in \SN$.
			Consider \REV{an }{}arbitrary $u$ \st $t \GT_\WPO u$.
			Since we have either $f \PGT g$ or
			$f \PGS g$ and 
			$\AppPerm{\sigma(f)}{s}{n} \GT_\WPO^\Lex \AppPerm{\sigma(g)}{t}{m}$,
			$\Tp{s,f,\AppPerm{\sigma(f)}{s}{n},|t|}$ is greater than
			$\Tp{t,g,\AppPerm{\sigma(g)}{t}{m},|u|}$ by
			\REV{the }{}second or third component\REV{}{s}.
			Hence, the induction hypothesis yields $u \in \SN$.\qedhere
		\end{enumerate}
	\end{enumerate}
\end{proof}

\begin{lemma}[Well-foundedness]\label{lem:WPO well-founded}
	If $\A$ is weakly simple \wrt $\sigma$,
	then $\GT_\WPO$ is well-founded.
\end{lemma}
\begin{proof}
	Let us show $s \in \SN$ for every term $s$
	by induction on $|s|$.
	\REV{The claim}{It} is trivial if $s \in \Vars$.
	Suppose $s = f(\Seq{s_n}) \GT_\WPO t$.
	By the induction hypothesis, we have $s_i \in \SN$ for every $i \in \SetOf{1,\dots,n}$.
	Hence by \prettyref{lem:WPO compose}, we get $t \in \SN$.\qedhere
\end{proof}

Now we prove that $\GS_\WPO$ and $\GT_\WPO$ are quasi\REV{-}{} and strict orders, \resp.

\begin{lemma}[Transitivity]\label{lem:WPO transitive}
	Both $\GS_\WPO$ and $\GT_\WPO$ are transitive.
\end{lemma}
\begin{proof}
	Analogous to \prettyref{lem:WPO compatible}.\qedhere
\end{proof}

\begin{lemma}[(Ir)reflexivity]
	$\GS_\WPO$ is reflexive and $\GT_\WPO$ is irreflexive.
\end{lemma}
\begin{proof}
	For arbitrary term $s$, $s \GS_\WPO s$ is easy by induction on $|s|$.
	Irreflexivity of $\GT_\WPO$ follows from
	\prettyref{lem:WPO well-founded}.\qedhere
\end{proof}

\REV{Now the}{Now} remaining properties required for a reduction pair are
\emph{stability} and \emph{weak monotonicity}, which are shown bellow.

\begin{lemma}[Stability]
	Both $\GS_\WPO$ and $\GT_\WPO$ are stable.
\end{lemma}
\begin{proof}
	Let $s = f(\Seq{s_n}) \GSopt_\WPO t$ and $\theta$ 
	\REV{be }{}an arbitrary substitution.
	We show $s\theta \GSopt_\WPO t\theta$ by induction on $|s| + |t|$.
	\REV{The claim}{It} is obvious if $s \AGT t$. Otherwise, we have
	$s \AGS t$ and obviously $s\theta \AGS t\theta$.
	The remaining cases are as follows:
	\begin{itemize}
	\item
		Suppose $s_i \GS_\WPO t$ for some $i \in \sigma(f)$.
		By the induction hypothesis, we get $s_i\theta \GS_\WPO t\theta$.
		Hence, case \prettyref{item:WPO-simp} applies for $s\theta \GSopt_\WPO t\theta$.
	\item
		Suppose $t = g(\Seq{t_m})$ and $s \GT_\WPO t_j$ for all $j \in \sigma(g)$.
		By the induction hypothesis, we get $s\theta \GT_\WPO t_j\theta$.
		\REV{The claim}{It} is obvious if $f \PGT g$. If $f \PSIM g$ and
		$\AppPerm{\sigma(f)}{s}{n} \GSopt_\WPO^\Lex \AppPerm{\sigma(g)}{t}{m}$,
		then by the induction hypothesis we get
		\[
			\AppPermSubst{\sigma(f)}{s}{n}{\theta} \GSopt_\WPO^\Lex
			\AppPermSubst{\sigma(g)}{t}{m}{\theta}
		\]
		Hence, case \prettyref{item:WPO-mono} applies for $s\theta \GSopt_\WPO t\theta$.
	\item
\REV{%
	If $s \GS_\WPO t$ is derived by either case \prettyref{item:WPO-min} or
	\prettyref{item:WPO-max}, then
	the claim follows from
	Proposition \ref{prop:least} or \ref{prop:greatest}.
}{}%
	\qedhere
	\end{itemize}
\end{proof}

\begin{lemma}[Weak monotonicity]
	If $\A$ is \REV{weakly monotone and }{}weakly simple \wrt $\sigma$,
	then $\GS_\WPO$ is monotone.
\end{lemma}
\begin{proof}
	Suppose $s_i \GS_\WPO s_i'$ and
	let us show $s = f(\dots,s_i,\dots) \GS_\WPO f(\dots,s_i',\dots) = s'$.
	Since $s_i \AGS s_i'$, we have
	$s \AGS s'$ by the weak monotonicity of $\A$.
	
	If $i \notin \sigma(f)$, then we have 
	$\AppPerm{\sigma(f)}{s}{n} = \AppPerm{\sigma(f)}{s'}{n}$.
	Otherwise, by the weak simplicity assumption we have
	$s \AGS s_j$ for every $j \in \sigma(f)$,
	and thus $s \GT_\WPO s_j$ by case \prettyref{item:WPO-simp} of \prettyref{def:WPOpS}.
	Hence case \prettyref{item:WPO-mono} applies for $s\GS_\WPO s'$.\qedhere
\end{proof}

\REV{%
The above results conclude \prettyref{thm:WPO pair}.
}{}%
\REV{%
	Now we prove \prettyref{thm:WPO simple}.
	Most of the required properties have already been obtained above,
	except for the following two:
}{%
	Now we show that $\WPO$ is a simplification order,
	and hence a reduction order.
	Most of the required properties will be obtained
	in \prettyref{sec:pair definition}, except for the following two:
}%

\begin{lemma}[\REV{Strict monotonicity}{}]\label{lem:WPO monotone}
	If $\A$ is weakly monotone and weakly simple,
	then $\GT_\WPO$ is monotone.
\end{lemma}
\begin{proof}
	Suppose $s_i \GT_\WPO s_i'$ and
	let us show $s = f(\dots,s_i,\dots) \REV{\GT_\WPO}{\GS_\WPO} f(\dots,s_i',\dots) = s'$.
	Since $s_i \AGS s_i'$, we have
	$s \AGS s'$ by the weak monotonicity of $\A$.
	By the weak simplicity of $\A$,
	we have $s \AGS s_j$ for every $j$,
	and thus $s \GT_\WPO s_j$ by case \prettyref{item:WPO-simp} of 
	\REV{\prettyref{def:WPO}}{\prettyref{def:WPOpS}}.
	Hence case \prettyref{item:WPO-mono} applies for $s \GT_\WPO s'$.\qedhere
\end{proof}

\begin{lemma}[\REV{Subterm property}{}]
	If $\A$ is weakly simple, then $\GT_\WPO$ has the subterm property.
\end{lemma}
\begin{proof}
	By the assumption, $f(\dots,s_i,\dots) \AGS s_i$ and hence
	$f(\dots,s_i,\dots) \GT_\WPO s_i$ by case \prettyref{item:WPO-simp}.
	\qedhere
\end{proof}

\bibliography{references}

\begin{thebibliography}{49}
\expandafter\ifx\csname natexlab\endcsname\relax\def\natexlab#1{#1}\fi
\providecommand{\url}[1]{\texttt{#1}}
\providecommand{\href}[2]{#2}
\providecommand{\path}[1]{#1}
\providecommand{\DOIprefix}{doi:}
\providecommand{\ArXivprefix}{arXiv:}
\providecommand{\URLprefix}{URL: }
\providecommand{\Pubmedprefix}{pmid:}
\providecommand{\doi}[1]{\href{http://dx.doi.org/#1}{\path{#1}}}
\providecommand{\Pubmed}[1]{\href{pmid:#1}{\path{#1}}}
\providecommand{\bibinfo}[2]{#2}
\ifx\xfnm\relax \def\xfnm[#1]{\unskip,\space#1}\fi
\bibitem[{Giesl et~al.(2006)Giesl, Schneider-Kamp, and Thiemann}]{GST06}
\bibinfo{author}{J.~Giesl}, \bibinfo{author}{P.~Schneider-Kamp},
  \bibinfo{author}{R.~Thiemann},
\newblock \bibinfo{title}{{AProVE} 1.2: Automatic termination proofs in the
  dependency pair framework},
\newblock in: \bibinfo{booktitle}{\Proceedings\IJCAR[3](2006)}, \LNAI[4130],
  \bibinfo{year}{2006}, pp. \bibinfo{pages}{281--286}.
\bibitem[{Korp et~al.(2009)Korp, Sternagel, Zankl, and Middeldorp}]{KSZM09}
\bibinfo{author}{M.~Korp}, \bibinfo{author}{C.~Sternagel},
  \bibinfo{author}{H.~Zankl}, \bibinfo{author}{A.~Middeldorp},
\newblock \bibinfo{title}{Tyrolean {T}ermination {T}ool~2},
\newblock in: \bibinfo{booktitle}{\Proceedings\RTA[20](2009)}, \LNCS[5595],
  \bibinfo{year}{2009}, pp. \bibinfo{pages}{295--304}.
\bibitem[{Kamin and L\'evy(1980)}]{KL80}
\bibinfo{author}{S.~Kamin}, \bibinfo{author}{J.-J. L\'evy}, \bibinfo{title}{Two
  generalizations of the recursive path ordering}, \bibinfo{year}{1980}.
  \bibinfo{note}{Unpublished note}.
\bibitem[{Dershowitz(1982)}]{D82}
\bibinfo{author}{N.~Dershowitz},
\newblock \bibinfo{title}{Orderings for term-rewriting systems},
\newblock \bibinfo{journal}{\TCS} \bibinfo{volume}{17} (\bibinfo{year}{1982})
  \bibinfo{pages}{279--301}.
\bibitem[{Lescanne(1983)}]{L83}
\bibinfo{author}{P.~Lescanne},
\newblock \bibinfo{title}{Computer experiments with the {REVE} term rewriting
  system generator},
\newblock in: \bibinfo{booktitle}{\Proceedings\POPL[10](1983)},
  \bibinfo{year}{1983}, pp. \bibinfo{pages}{99--108}.
\bibitem[{Codish et~al.(2012)Codish, Giesl, Schneider-Kamp, and
  Thiemann}]{CGST12}
\bibinfo{author}{M.~Codish}, \bibinfo{author}{J.~Giesl},
  \bibinfo{author}{P.~Schneider-Kamp}, \bibinfo{author}{R.~Thiemann},
\newblock \bibinfo{title}{{SAT} solving for termination proofs with recursive
  path orders and dependency pairs},
\newblock \bibinfo{journal}{\JAR} \bibinfo{volume}{49} (\bibinfo{year}{2012})
  \bibinfo{pages}{53--93}.
\bibitem[{Knuth and Bendix(1970)}]{KB70}
\bibinfo{author}{D.~Knuth}, \bibinfo{author}{P.~Bendix},
\newblock \bibinfo{title}{Simple word problems in universal algebras},
\newblock in: \bibinfo{booktitle}{Computational Problems in Abstract Algebra},
  \bibinfo{publisher}{Pergamon Press}, \bibinfo{address}{New York},
  \bibinfo{year}{1970}, pp. \bibinfo{pages}{263--297}.
\bibitem[{Korovin and Voronkov(2003)}]{KV03}
\bibinfo{author}{K.~Korovin}, \bibinfo{author}{A.~Voronkov},
\newblock \bibinfo{title}{Orienting rewrite rules with the {K}nuth-{B}endix
  order},
\newblock \bibinfo{journal}{\InfComput} \bibinfo{volume}{183}
  (\bibinfo{year}{2003}) \bibinfo{pages}{165--186}.
\bibitem[{Zankl et~al.(2009)Zankl, Hirokawa, and Middeldorp}]{ZHM09}
\bibinfo{author}{H.~Zankl}, \bibinfo{author}{N.~Hirokawa},
  \bibinfo{author}{A.~Middeldorp},
\newblock \bibinfo{title}{{KBO} orientability},
\newblock \bibinfo{journal}{\JAR} \bibinfo{volume}{43} (\bibinfo{year}{2009})
  \bibinfo{pages}{173--201}.
\bibitem[{Middeldorp and Zantema(1997)}]{MZ97}
\bibinfo{author}{A.~Middeldorp}, \bibinfo{author}{H.~Zantema},
\newblock \bibinfo{title}{Simple termination of rewrite systems},
\newblock \bibinfo{journal}{\TCS} \bibinfo{volume}{175} (\bibinfo{year}{1997})
  \bibinfo{pages}{127--158}.
\bibitem[{Ludwig and Waldmann(2007)}]{LW07}
\bibinfo{author}{M.~Ludwig}, \bibinfo{author}{U.~Waldmann},
\newblock \bibinfo{title}{An extension of the {K}nuth-{B}endix ordering with
  {LPO}-like properties},
\newblock in: \bibinfo{booktitle}{\Proceedings\LPAR[14](2007)}, \LNAI[4790],
  \bibinfo{year}{2007}, pp. \bibinfo{pages}{348--362}.
\bibitem[{Kov{\'a}cs et~al.(2011)Kov{\'a}cs, Moser, and Voronkov}]{KMV11}
\bibinfo{author}{L.~Kov{\'a}cs}, \bibinfo{author}{G.~Moser},
  \bibinfo{author}{A.~Voronkov},
\newblock \bibinfo{title}{On transfinite {K}nuth-{B}endix orders},
\newblock in: \bibinfo{booktitle}{\Proceedings\CADE[23](2011)}, \LNAI[6803],
  \bibinfo{year}{2011}, pp. \bibinfo{pages}{384--399}.
\bibitem[{Winkler et~al.(2012)Winkler, Zankl, and Middeldorp}]{WZM12}
\bibinfo{author}{S.~Winkler}, \bibinfo{author}{H.~Zankl},
  \bibinfo{author}{A.~Middeldorp},
\newblock \bibinfo{title}{Ordinals and {K}nuth-{B}endix orders},
\newblock in: \bibinfo{booktitle}{\Proceedings\LPAR[18](2012)}, \ARCoSS[7180],
  \bibinfo{year}{2012}, pp. \bibinfo{pages}{420--434}.
\bibitem[{Lankford(1975)}]{L75}
\bibinfo{author}{D.~Lankford}, \bibinfo{title}{Canonical algebraic
  simplification in computational logic}, \bibinfo{type}{Technical Report}
  \bibinfo{number}{ATP-25}, University of Texas, \bibinfo{year}{1975}.
\bibitem[{Zantema(2001)}]{Z01}
\bibinfo{author}{H.~Zantema},
\newblock \bibinfo{title}{The termination hierarchy for term rewriting},
\newblock \bibinfo{journal}{\AAECC} \bibinfo{volume}{12} (\bibinfo{year}{2001})
  \bibinfo{pages}{3--19}.
\bibitem[{Fuhs et~al.(2008)Fuhs, Giesl, Middeldorp, Schneider-Kamp, Thiemann,
  and Zankl}]{FGMSTZ08}
\bibinfo{author}{C.~Fuhs}, \bibinfo{author}{J.~Giesl},
  \bibinfo{author}{A.~Middeldorp}, \bibinfo{author}{P.~Schneider-Kamp},
  \bibinfo{author}{R.~Thiemann}, \bibinfo{author}{H.~Zankl},
\newblock \bibinfo{title}{Maximal termination},
\newblock in: \bibinfo{booktitle}{\Proceedings\RTA[19](2008)}, \LNCS[5117],
  \bibinfo{year}{2008}, pp. \bibinfo{pages}{110--125}.
\bibitem[{Fuhs et~al.(2007)Fuhs, Giesl, Middeldorp, Schneider-Kamp, Thiemann,
  and Zankl}]{FGMSTZ07}
\bibinfo{author}{C.~Fuhs}, \bibinfo{author}{J.~Giesl},
  \bibinfo{author}{A.~Middeldorp}, \bibinfo{author}{P.~Schneider-Kamp},
  \bibinfo{author}{R.~Thiemann}, \bibinfo{author}{H.~Zankl},
\newblock \bibinfo{title}{{SAT} solving for termination analysis with
  polynomial interpretations},
\newblock in: \bibinfo{booktitle}{\Proceedings\SAT[10](2007)}, \LNCS[4501],
  \bibinfo{year}{2007}, pp. \bibinfo{pages}{340--354}.
\bibitem[{Arts and Giesl(2000)}]{AG00}
\bibinfo{author}{T.~Arts}, \bibinfo{author}{J.~Giesl},
\newblock \bibinfo{title}{Termination of term rewriting using dependency
  pairs},
\newblock \bibinfo{journal}{\TCS} \bibinfo{volume}{236} (\bibinfo{year}{2000})
  \bibinfo{pages}{133--178}.
\bibitem[{Hirokawa and Middeldorp(2005)}]{HM05}
\bibinfo{author}{N.~Hirokawa}, \bibinfo{author}{A.~Middeldorp},
\newblock \bibinfo{title}{Automating the dependency pair method},
\newblock \bibinfo{journal}{\InfComput} \bibinfo{volume}{199}
  (\bibinfo{year}{2005}) \bibinfo{pages}{172--199}.
\bibitem[{Giesl et~al.(2006)Giesl, Thiemann, Schneider-Kamp, and
  Falke}]{GTSF06}
\bibinfo{author}{J.~Giesl}, \bibinfo{author}{R.~Thiemann},
  \bibinfo{author}{P.~Schneider-Kamp}, \bibinfo{author}{S.~Falke},
\newblock \bibinfo{title}{Mechanizing and improving dependency pairs},
\newblock \bibinfo{journal}{\JAR} \bibinfo{volume}{37} (\bibinfo{year}{2006})
  \bibinfo{pages}{155--203}.
\bibitem[{Endrullis et~al.(2008)Endrullis, Waldmann, and Zantema}]{EWZ08}
\bibinfo{author}{J.~Endrullis}, \bibinfo{author}{J.~Waldmann},
  \bibinfo{author}{H.~Zantema},
\newblock \bibinfo{title}{Matrix interpretations for proving termination of
  term rewriting},
\newblock \bibinfo{journal}{\JAR} \bibinfo{volume}{40} (\bibinfo{year}{2008})
  \bibinfo{pages}{195--220}.
\bibitem[{Bofill et~al.(2013)Bofill, Borralleras, Rodr\'iguez-Carbonell, and
  Rubio}]{BBRR13}
\bibinfo{author}{M.~Bofill}, \bibinfo{author}{C.~Borralleras},
  \bibinfo{author}{E.~Rodr\'iguez-Carbonell}, \bibinfo{author}{A.~Rubio},
\newblock \bibinfo{title}{The recursive path and polynomial ordering for
  first-order and higher-order terms},
\newblock \bibinfo{journal}{\JLC} \bibinfo{volume}{23} (\bibinfo{year}{2013})
  \bibinfo{pages}{263--305}.
\bibitem[{Borralleras et~al.(2000)Borralleras, Ferreira, and Rubio}]{BFR00}
\bibinfo{author}{C.~Borralleras}, \bibinfo{author}{M.~Ferreira},
  \bibinfo{author}{A.~Rubio},
\newblock \bibinfo{title}{Complete monotonic semantic path orderings},
\newblock in: \bibinfo{booktitle}{\Proceedings\CADE[17](2000)}, \LNCS[1831],
  \bibinfo{year}{2000}, pp. \bibinfo{pages}{346--364}.
\bibitem[{Dershowitz and Hoot(1995)}]{DH95}
\bibinfo{author}{N.~Dershowitz}, \bibinfo{author}{C.~Hoot},
\newblock \bibinfo{title}{Natural termination},
\newblock \bibinfo{journal}{\TCS} \bibinfo{volume}{142} (\bibinfo{year}{1995})
  \bibinfo{pages}{179--207}.
\bibitem[{Geser(1996)}]{G96}
\bibinfo{author}{A.~Geser},
\newblock \bibinfo{title}{An improved general path order},
\newblock \bibinfo{journal}{\AAECC} \bibinfo{volume}{7} (\bibinfo{year}{1996})
  \bibinfo{pages}{469--511}.
\bibitem[{Yamada et~al.(2013)Yamada, Kusakari, and Sakabe}]{YKS13b}
\bibinfo{author}{A.~Yamada}, \bibinfo{author}{K.~Kusakari},
  \bibinfo{author}{T.~Sakabe},
\newblock \bibinfo{title}{Partial status for {KBO}},
\newblock in: \bibinfo{booktitle}{\Proceedings\WST[13](2013)},
  \bibinfo{year}{2013}, pp. \bibinfo{pages}{74--78}.
\bibitem[{TPDB(2013)}]{TPDB13}
TPDB, \bibinfo{title}{The termination problem data base, version 8.0.6},
  \bibinfo{year}{2013}. \URLprefix
  \url{http://termination-portal.org/wiki/TPDB}.
\bibitem[{Yamada et~al.(2013)Yamada, Kusakari, and Sakabe}]{YKS13}
\bibinfo{author}{A.~Yamada}, \bibinfo{author}{K.~Kusakari},
  \bibinfo{author}{T.~Sakabe},
\newblock \bibinfo{title}{Unifying the {K}nuth-{B}endix, recursive path and
  polynomial orders},
\newblock in: \bibinfo{booktitle}{\Proceedings\PPDP[15](2013)},
  \bibinfo{year}{2013}, pp. \bibinfo{pages}{181--192}.
\bibitem[{Baader and Nipkow(1998)}]{BN98}
\bibinfo{author}{F.~Baader}, \bibinfo{author}{T.~Nipkow}, \bibinfo{title}{Term
  Rewriting and All That}, \bibinfo{publisher}{Cambridge University Press},
  \bibinfo{year}{1998}.
\bibitem[{TeReSe(2003)}]{Terese}
\bibinfo{author}{TeReSe}, \bibinfo{title}{Term Rewriting Systems}, \CTTCS[55],
  \bibinfo{publisher}{Cambridge University Press}, \bibinfo{year}{2003}.
\bibitem[{Zantema(1994)}]{Z94}
\bibinfo{author}{H.~Zantema},
\newblock \bibinfo{title}{Termination of term rewriting: interpretation and
  type elimination},
\newblock \bibinfo{journal}{\JSC} \bibinfo{volume}{17} (\bibinfo{year}{1994})
  \bibinfo{pages}{23--50}.
\bibitem[{Manna and Ness(1970)}]{MN70}
\bibinfo{author}{Z.~Manna}, \bibinfo{author}{S.~Ness},
\newblock \bibinfo{title}{On the termination of {M}arkov algorithms},
\newblock in: \bibinfo{booktitle}{\Proceedings{the 3rd Hawaii International
  Conference on System Science}}, \bibinfo{year}{1970}, pp.
  \bibinfo{pages}{789--792}.
\bibitem[{Steinbach(1989)}]{S89}
\bibinfo{author}{J.~Steinbach},
\newblock \bibinfo{title}{Extensions and comparison of simplification orders},
\newblock in: \bibinfo{booktitle}{\Proceedings\RTA[3](1989)}, \LNCS[355],
  \bibinfo{year}{1989}, pp. \bibinfo{pages}{434--448}.
\bibitem[{Giesl et~al.(2004)Giesl, Thiemann, and Schneider-Kamp}]{GTS04}
\bibinfo{author}{J.~Giesl}, \bibinfo{author}{R.~Thiemann},
  \bibinfo{author}{P.~Schneider-Kamp},
\newblock \bibinfo{title}{The dependency pair framework: Combining techniques
  for automated termination proofs},
\newblock in: \bibinfo{booktitle}{\Proceedings\LPAR[11](2004)}, \LNAI[3452],
  \bibinfo{year}{2004}, pp. \bibinfo{pages}{75--90}.
\bibitem[{Kusakari et~al.(1999)Kusakari, Nakamura, and Toyama}]{KNT00}
\bibinfo{author}{K.~Kusakari}, \bibinfo{author}{M.~Nakamura},
  \bibinfo{author}{Y.~Toyama},
\newblock \bibinfo{title}{Argument filtering transformation},
\newblock in: \bibinfo{booktitle}{\Proceedings\PPDP[1](1999)}, \LNCS[1702],
  \bibinfo{year}{1999}, pp. \bibinfo{pages}{47--61}.
\bibitem[{Koprowski and Waldmann(2009)}]{KW09}
\bibinfo{author}{A.~Koprowski}, \bibinfo{author}{J.~Waldmann},
\newblock \bibinfo{title}{Max/plus tree automata for termination of term
  rewriting},
\newblock \bibinfo{journal}{\ActaCybern} \bibinfo{volume}{19}
  (\bibinfo{year}{2009}) \bibinfo{pages}{357--392}.
\bibitem[{Sternagel and Thiemann(2013)}]{ST13}
\bibinfo{author}{C.~Sternagel}, \bibinfo{author}{R.~Thiemann},
\newblock \bibinfo{title}{Formalizing {K}nuth-{B}endix orders and
  {K}nuth-{B}endix completion},
\newblock in: \bibinfo{booktitle}{\Proceedings\RTA[24](2013)}, \LIPIcs[21],
  \bibinfo{year}{2013}, pp. \bibinfo{pages}{287--302}.
\bibitem[{Thiemann et~al.(2012)Thiemann, Allais, and Nagele}]{TAN12}
\bibinfo{author}{R.~Thiemann}, \bibinfo{author}{G.~Allais},
  \bibinfo{author}{J.~Nagele},
\newblock \bibinfo{title}{On the formalization of termination techniques based
  on multiset orderings},
\newblock in: \bibinfo{booktitle}{\Proceedings\RTA[23](2012)}, \LIPIcs[15],
  \bibinfo{year}{2012}, pp. \bibinfo{pages}{339--354}.
\bibitem[{Korovin and Voronkov(2003)}]{KV03b}
\bibinfo{author}{K.~Korovin}, \bibinfo{author}{A.~Voronkov},
\newblock \bibinfo{title}{An {AC}-compatible {K}nuth-{B}endix order},
\newblock in: \bibinfo{booktitle}{\Proceedings\CADE[19](2003)}, \LNAI[2741],
  \bibinfo{year}{2003}, pp. \bibinfo{pages}{47--59}.
\bibitem[{Ben-Amram and Codish(2008)}]{BC08}
\bibinfo{author}{A.~M. Ben-Amram}, \bibinfo{author}{M.~Codish},
\newblock \bibinfo{title}{A {SAT}-based approach to size change termination
  with global ranking functions},
\newblock in: \bibinfo{booktitle}{\Proceedings\TACAS[14](2008)}, \LNCS[4963],
  \bibinfo{year}{2008}, pp. \bibinfo{pages}{218--232}.
\bibitem[{Giesl et~al.(2005)Giesl, Thiemann, and Schneider-Kamp}]{GTS05}
\bibinfo{author}{J.~Giesl}, \bibinfo{author}{R.~Thiemann},
  \bibinfo{author}{P.~Schneider-Kamp},
\newblock \bibinfo{title}{Proving and disproving termination of higher-order
  functions},
\newblock in: \bibinfo{booktitle}{\Proceedings\FroCoS[5](2005)}, \LNAI[3717],
  \bibinfo{year}{2005}, pp. \bibinfo{pages}{216--231}.
\bibitem[{Codish et~al.(2006)Codish, Schneider-Kamp, Lagoon, Thiemann, and
  Giesl}]{CSLTG06}
\bibinfo{author}{M.~Codish}, \bibinfo{author}{P.~Schneider-Kamp},
  \bibinfo{author}{V.~Lagoon}, \bibinfo{author}{R.~Thiemann},
  \bibinfo{author}{J.~Giesl},
\newblock \bibinfo{title}{{SAT} solving for argument filterings},
\newblock in: \bibinfo{booktitle}{\Proceedings\LPAR[13](2006)}, \LNCS[4246],
  \bibinfo{year}{2006}, pp. \bibinfo{pages}{30--44}.
\bibitem[{TermComp(2013)}]{TC13}
TermComp, \bibinfo{title}{The termination competition}, \bibinfo{year}{2013}.
  \URLprefix \url{http://termcomp.uibk.ac.at/termcomp/}.
\bibitem[{Zankl and Middeldorp(2010)}]{ZM10}
\bibinfo{author}{H.~Zankl}, \bibinfo{author}{A.~Middeldorp},
\newblock \bibinfo{title}{Satisfiability of non-linear (ir)rational
  arithmetic},
\newblock in: \bibinfo{booktitle}{\Proceedings\LPAR[16](2010)}, \LNAI[6355],
  \bibinfo{year}{2010}, pp. \bibinfo{pages}{481--500}.
\bibitem[{Borralleras et~al.(2012)Borralleras, Lucas, Navarro-Marset,
  Rodr\'{\i}guez-Carbonell, and Rubio}]{BLNRR12}
\bibinfo{author}{C.~Borralleras}, \bibinfo{author}{S.~Lucas},
  \bibinfo{author}{R.~Navarro-Marset},
  \bibinfo{author}{E.~Rodr\'{\i}guez-Carbonell}, \bibinfo{author}{A.~Rubio},
\newblock \bibinfo{title}{{SAT} modulo linear arithmetic for solving polynomial
  constraints},
\newblock \bibinfo{journal}{\JAR} \bibinfo{volume}{48} (\bibinfo{year}{2012})
  \bibinfo{pages}{107--131}.
\bibitem[{Hirokawa et~al.(2013)Hirokawa, Middeldorp, and Zankl}]{HMZ13}
\bibinfo{author}{N.~Hirokawa}, \bibinfo{author}{A.~Middeldorp},
  \bibinfo{author}{H.~Zankl},
\newblock \bibinfo{title}{Uncurrying for termination and complexity},
\newblock \bibinfo{journal}{\JAR} \bibinfo{volume}{50} (\bibinfo{year}{2013})
  \bibinfo{pages}{279--315}.
\bibitem[{Sternagel and Thiemann(2011)}]{ST11}
\bibinfo{author}{C.~Sternagel}, \bibinfo{author}{R.~Thiemann},
\newblock \bibinfo{title}{Generalized and formalized uncurrying},
\newblock in: \bibinfo{booktitle}{\Proceedings\FroCoS[8](2011)}, \LNAI[6989],
  \bibinfo{year}{2011}, pp. \bibinfo{pages}{243--258}.
\bibitem[{Yamada et~al.(2014)Yamada, Kusakari, and Sakabe}]{YKS14b}
\bibinfo{author}{A.~Yamada}, \bibinfo{author}{K.~Kusakari},
  \bibinfo{author}{T.~Sakabe},
\newblock \bibinfo{title}{{N}agoya {T}ermination {T}ool},
\newblock in: \bibinfo{booktitle}{\Proceedings Joint 25th International
  Conference on Rewriting Techniques and Applications and 12th International
  Conference on Typed Lambda Calculi and Applications (RTA-TLCA\;'14)}, \LNCS,
  \bibinfo{year}{2014}. \bibinfo{note}{To appear}.
\bibitem[{Winkler et~al.(2013)Winkler, Zankl, and Middeldorp}]{WZM13}
\bibinfo{author}{S.~Winkler}, \bibinfo{author}{H.~Zankl},
  \bibinfo{author}{A.~Middeldorp},
\newblock \bibinfo{title}{Beyond {P}eano arithmetic -- automatically proving
  termination of the {G}oodstein sequence},
\newblock in: \bibinfo{booktitle}{\Proceedings\RTA[24](2013)}, \LIPIcs[21],
  \bibinfo{year}{2013}, pp. \bibinfo{pages}{335--351}.

\end{thebibliography}

\end{document}